\documentclass[12pt]{article}
\usepackage{amsmath}
\usepackage{amssymb}
\usepackage{graphicx,psfrag,epsf}
\usepackage{enumerate}
\usepackage{amsthm}
\usepackage[sort]{natbib}%
\usepackage{multirow}
\usepackage{caption,subcaption}

\usepackage{setspace}
\newcommand{\blind}{0}

\def\d{\mbox{d}}

\def\half{\hbox{$1\over2$}}

\def\d{{\rm d}}

\newtheorem{theorem}{Theorem}[section]
\newtheorem{lemma}{Lemma}

\numberwithin{equation}{section}

\graphicspath{{Figures/}}

\begin{document}



\def\spacingset#1{\renewcommand{\baselinestretch}%
{#1}\small\normalsize} \spacingset{1}

\if0\blind
{
\title{\bf An Objective Prior from a Scoring Rule }
\author{Stephen G. Walker\footnote{Department of Mathematics,
University of Texas at Austin, USA. 
e-mail: s.g.walker@math.utexas.edu}\hspace{.2cm} \& \hspace{.2cm} Cristiano Villa\footnote{
School of Mathematics, Statistics \& Physics, University of Newcastle, UK. e-mail: Cristiano.Villa@ncl.ac.uk}
}
\date{}
\maketitle       
} \fi   

\if1\blind
{
  \bigskip
  \bigskip
  \bigskip
  \begin{center}
    {\LARGE\bf }
\end{center}
  \medskip
} \fi




\bigskip
\begin{abstract}
In this paper we introduce a novel objective prior distribution levering on the connections between information, divergence and scoring rules. In particular, we do so from the starting point of convex functions representing information in density functions. This provides a natural route to proper local scoring rules using Bregman divergence. Specifically, we determine the prior which solves setting the score function to be a constant. While in itself this provides motivation for an objective prior, the prior also minimizes a corresponding information criterion.

\end{abstract}








\section{Introduction}\label{sc_introduction}
A major drawback of objective priors, such as Jeffreys prior \citep{Jeff1961} and the reference prior \citep{Bern:1979}, is that, in many cases, they are improper. While for a parameter that is defined over a bounded interval, such as $(0,1)$, it is generally possible to derive objective prior distributions that are proper, this is not the case for parameters on $(0,\infty)$ or $(-\infty,\infty)$. The literature provides many examples where improper prior distributions cannot be suitably employed; such as Bayes factors, mixture models and hierarchical models, to name but a few. Methods have been proposed to get around these obstacles, for example, Intrinsic Bayes Factors
\citep{BergPer:1996} and Fractional Bayes factors \citep{Ohag:1995} or reparametrising mixture models \citep{GrazianRobert2018}. However, these type of results are generally valid for a limited number of  specific conditions. Additonally, improper prior distributions are not too suitable to be employed where large numbers of parameters are involved as it would be difficult to establish properness of the full posterior distribution.

The idea of this paper is to present a novel objective prior distribution for continuous parameter spaces by considering the connection between  information, divergence and scoring rules.  In particular, the proposed prior can be defined over $(0,\infty)$ and $(-\infty,\infty)$, the latter by extending the former, and it has the appealing property of being proper. 

Recently, \cite{LVW:2020}, introduced a new class of objective prior which solved a differential equation of the form
$S(q,q',q'')=0$, where $S$ is a score function and the solution $q$ acts as the prior distribution. The solution is also shown to minimize an information criterion. 

There are two well known relations that connect information, proper local scoring rules and divergences. The most famous of which links Shannon information, Kullback--Leibler divergence and the log--score, given by
\begin{equation}\label{dis}
\int p\log \frac{p}{q}=\int p\log p +\int p\,(-\log q),
\end{equation}
where $p$ and $q$ are two densities,
and integrals will be generally defined with respect to the Lebesgue measure.
The term on the left-hand-side of \eqref{dis} is the Kullback--Leibler divergence \citep{Kull:1951} between $p$ and $q$, the first term on the right-hand-side is the Shannon information associated with density $p$, and the second term is the expectation of the log-score function.

Another way to connect information, divergence and proper local scoring rules, involves Fisher divergence, Fisher information, and the Hyv\"{a}rinen score function (\cite{Hyva:2005}):
$$\int p\left(\frac{p'}{p}-\frac{q'}{q}\right)^2=\int \frac{(p')^2}{p} +\int p\left(2\frac{q''}{q}-\left(\frac{q'}{q}\right)^2\right),$$
where the final term has been obtained using an integration by parts.
In general, these relationships can be expressed as
\begin{equation}\label{generalrelation}
D(p,q)=I(p)+\int p\,S(q),
\end{equation}
where $D$ denotes the divergence, $I$ the measure of information and $S$ the score.

Recently, in \cite{Parry:2012}, a new class of score function was introduced, where the starting point is the property of the score function, which is that
$$p=\arg\min_q\int p\,S(q),$$
for all densities $p$. In other words, a score is said to be proper if the above is minimised by the choice of $q=p$.
Let us consider the well known log-score, $S(q)=-\log q(x)$. Then, we have that it satisfies the above property, since for any density $p$ it is that
$\int p\,\log(p/q)\geq 0$,
with equality only when $q\equiv p$. As such, we have that the log-score is a proper score. Furthermore, a score is said to be local if it only depends on $q$ through the density value $q(x)$. See \cite{Parry:2012} and \cite{EhmGnei:2012}. It has to be noted that the log-score is the only proper score to be local.

If we consider the Hyv\"{a}rinen score function \citep{Hyva:2005}, which is given by
$$S(q,q',q'')=2\,\frac{q''}{q}-\left(\frac{q'}{q}\right)^2,$$
which we note depends on $(q,q',q'')$, i.e. $q$ and the first two derivatives, as such it is not local in the above sense. However, the locality condition can be weakened \citep{Parry:2012} by allowing the score to depend on a finite number $m$ of derivatives. Therefore, the Hyv\"{a}rinen score will be a order--2 proper local scoring rule.

More generally, if a proper score depends on $m$ derivatives, then it will be defined an \textit{order--m local scoring rule}.
The theory in support of this, is based on the fact that the minimizer of  $\int p\,S(q)$ is $p$, and this can be investigated using variational analysis.
The relevant Euler--Lagrange equation of order two being
\begin{equation}\label{euler}
S+q\frac{\partial S}{\partial q}-\frac{d}{dx}\,q\,\frac{\partial S}{\partial q'}+\frac{d^2}{dx^2}\,q\,\frac{\partial S}{\partial q''}=0.
\end{equation}
The corresponding general case of \eqref{euler} is given as equation (18) in \cite{Parry:2012}. 
Throughout this paper we will focus on the case $m=2$, since this is where we draw our prior from.  
The Appendix provides the expression for a general $m$.

In \cite{Parry:2012}, the solution to equation \eqref{euler}, is proposed using properties of differential operators and $1$--homogeneous functions. Recall that a $1$--homogeneous function $f$ is such that $f(x,\lambda q,\lambda q')=\lambda\,f(x,q,q')$ for any $\lambda>0$.  
In particular, the Hyv\"{a}rinen score arises with $f(x,q,q')=(q')^2/q$ and 
\begin{equation}\label{hyvar}
S(q)=-\frac{\partial f}{\partial q}+\frac{d}{dx}\frac{\partial f}{\partial q'}.
\end{equation}
Furthermore, \cite{Parry:2012} and \cite{EhmGnei:2012} characterize all local and proper scoring rules of order $m=2$. 
With this respect, as an additional interesting result, in the Appendix we present the characterization using measures of information and the Bregman divergence \citep{Breg:1967}. The benefits of the proposed approach are that complicated mathematical analysis is avoided and the derivation of the local rule is made explicit. 

Following \cite{Parry:2012} and \cite{EhmGnei:2012} and the novel derivation of their results using Bregman divergence, which is the focus of the Appendix, information, divergence and scores can be obtained as follows: For some convex function $\alpha:\mathbb{R}\to\mathbb{R}$, 

\begin{description}

\item 1. \textit{Divergence}: Given the result \eqref{cond}, we get
$$D(p,q)=\int p\,\alpha(p'/p)-\int p\,\frac{\partial \phi}{\partial q}-\int p'\,\frac{\partial \phi}{\partial q'},$$
where
$$\frac{\partial \phi}{\partial q}=\alpha(q'/q)-(q'/q)\,\alpha'(q'/q)\quad\mbox{and}\quad \frac{\partial \phi}{\partial q'}=\alpha'(q'/q).$$
Using integration by parts on the right most intergal, and assuming that $[p\,\cdot\,\partial\phi/\partial q']$ vanishes at the extremes of the integral,  
$$D(p,q)=\int p\,\alpha(p'/p)+\int p\left\{\frac{d}{dx}\alpha'(q'/q)-\alpha(q'/q)+(q'/q)\,\alpha'(q'/q)\right\}. $$

\item 2. \textit{Information}: This follows from the divergence, and from \eqref{divinfsc}, and is given by
$$I(p)=\int p\,\alpha(p'/p).$$

\item 3. \textit{Score}: Again, from the form of the divergence and \eqref{score}, this is given by
\begin{equation}\label{ascore}
S(q,q',q'')=\frac{d}{dx}\alpha'(q'/q)-\alpha(q'/q)+(q'/q)\,\alpha'(q'/q).
\end{equation}

\end{description}

\noindent
The score $S(q,q',q'')$ in \eqref{ascore} generalizes the Hyv\"{a}rinen score, which arises when $\alpha(u)=u^2$.

\vspace{0.2in}
\noindent
The paper is organised as follows. Section \ref{sc_application} introduces the proposed objective prior. Section \ref{sc_mixtures} includes  a thorough simulation study, and an application to mixture models that involves both simulated and real data.  In Section \ref{sc_hierarchical} we have discussed another critical scenario where improper priors may resolve in improper posteriors, that is assigning an objective prior to the variance parameter in a hierarchical model.
The supporting theory is presented in the Appendix. In \ref{sc_informationandbregman} we use Bregman divergence to obtain general forms for score functions and associated divergences and following on from this
in \ref{sc_divergences} we detail  how we use Bregman divergences to obtain a divergence between probability density functions using their first derivatives, and show how to obtain score functions from these divergences. \ref{sc_highorder} provides the general case using $m$ derivatives. Finally, in \ref{ParryApp} we make the connection with our derivations of scores and that of \cite{Parry:2012}.

\section{New Objective prior}\label{sc_application}

\cite{LVW:2020} proposed constructing objective prior distributions on parameter spaces by solving equations of the kind
$S(q)=0$. Specifically, they used a weighted mixture of the log-score and the Hyv\"{a}rinen score functions. Note that the sole use of the log-score function would result in the uniform prior, which is not appropriate in many cases and may yield improper posterior distributions. On the other hand, a weighted combination of the two score functions yields a differential equation given by
\begin{equation}\label{LVWdiffeq}
-w\log q(x)+\frac{q''(x)}{q(x)}-\half \left(\frac{q'(x)}{q(x)}\right)^2=0,
\end{equation}
where $q$ denotes the prior density and $w$ the weight balancing the two score functions.
 
Solutions to the differential equation \eqref{LVWdiffeq} can be found for different spaces, and constraints on the shape of $q$ can be considered; so to have a prior density with desirable behaviour, such as monotone, convex, log--concave and more.

We have already seen that the Hyv\"{a}rinen score arises with $\alpha(u)=u^2$; see \eqref{ascore}. 
An important property an objective  prior distribution may be required to have is a heavy tail. We will consider such on $(0,\infty)$. 
Mirroring the Hyv\"{a}rinen score, we adopt $\alpha(u)=u^{-2}$ with $u=-q'/q$, and $q$ a decreasing density on $(0,\infty)$. In this case, equation \eqref{ascore} becomes
$6\,u^\prime/u^4 -3/u^2,$
which, by setting to 0, becomes
$u^\prime = \half u^2.$
The solution is easily seen to be
$u(x) = -2/(a+x),$
for some constant $a$. In this case, the prior on the parameter space $(0,\infty)$ is
\begin{equation}\label{eq_prior2}
q(x) = \frac{a}{(a+x)^2}.
\end{equation}
Interestingly, the prior in \eqref{eq_prior2}, is a Lomax distribution \citep{Lomax:1954} with scale parameter $a$ and shape parameter equal to 1. Recalling that the Lomax distribution can be directly connected to the Pareto Type I and Pareto Type II distributions, its heavy-tailed nature is immediately obvious. 

Fig.~\ref{fig:prior2} shows the prior with $a=1$. 
\begin{figure}[h!]
\centering
\includegraphics[width=14cm,height=6cm]{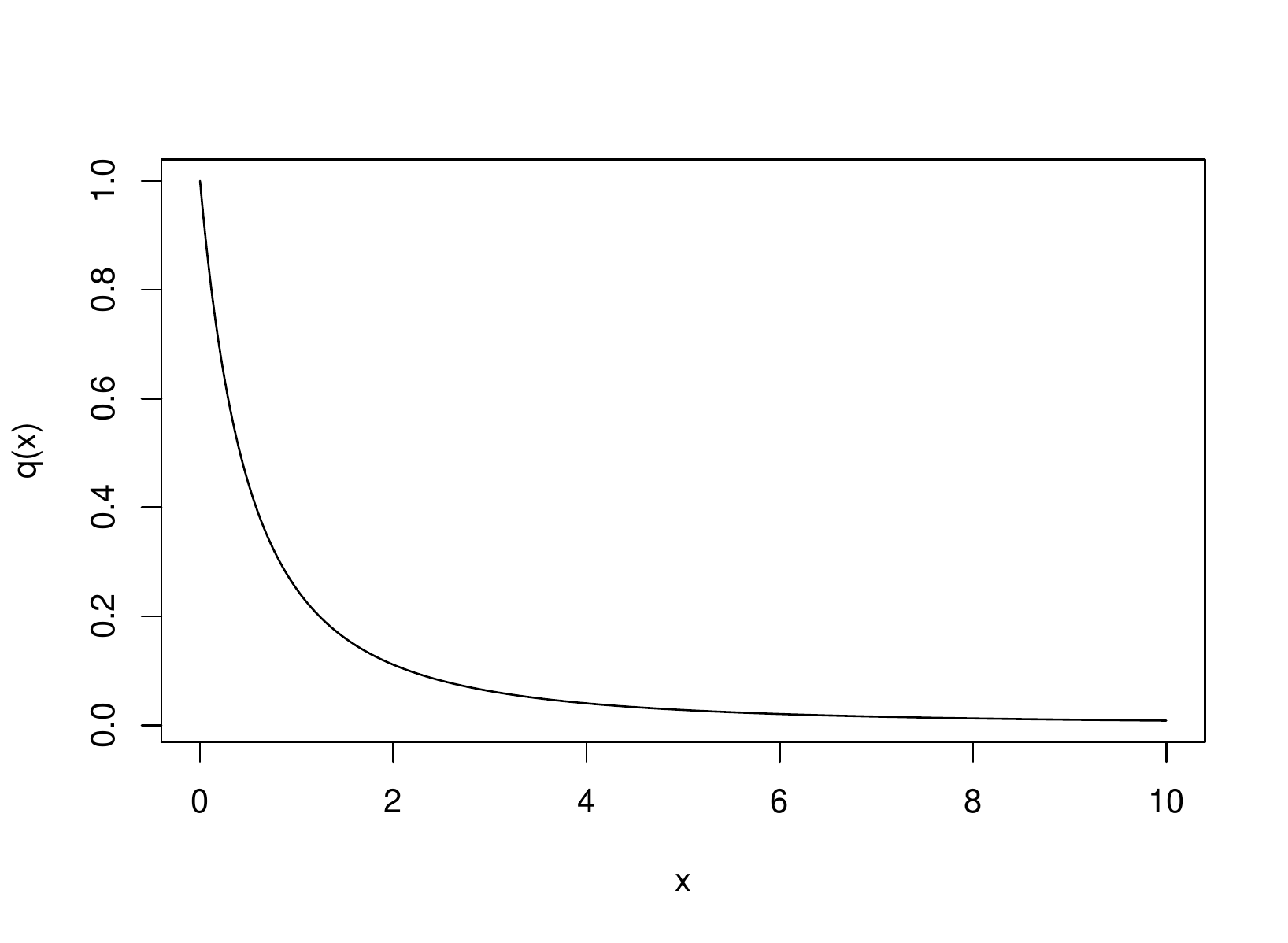}
\caption{Prior $q(x)=1/(1+x)^2$}
\label{fig:prior2}
\end{figure}

Making the connection more directly with the theory set out in the paper  with $\alpha(u)=u^{-2}$, we have
$\phi(u,v)=u^3/v^2$ which is easy to show satisfies $\phi=u\partial\phi/\partial u+v\partial \phi/\partial v$. Then  using (\ref{score}) we get
\begin{equation}\label{2dimscore}
S(q,q',q'')=3\left(\frac{q}{q'}\right)^2\left\{2\frac{q\,q''}{(q')^2}-3\right\}.
\end{equation}
Setting this to zero; i.e. $2qq''=3(q')^2$, this
can be solved and the solution is precisely of the form $a/(a+x)^2$. We now write this all out in a theorem.

\begin{theorem}
Let $\phi(p,p')=p^3/(p')^2$ be the convex function appearing in (\ref{breg2dim}); i.e. $\phi(u,v)=u\alpha(v/u)$ with $\alpha(\xi)=\xi^{-2}$. Then $\phi$ is convex for either $\xi<0$ or $\xi>0$. The Euler equation associated with this $\phi$; i.e. $d/dx\,\, \partial \phi/\partial p'=\partial \phi/\partial p$, yields
$$6p^3\,p''/(p')^4-6(p/p')^2=3(p/p')^2,$$
the solution to which can be written as $S(p,p',p'')=0$ where $S$ is the corresponding score function (\ref{2dimscore}). 
\end{theorem}

To obtain the corresponding prior on $(-\infty,\infty)$ through symmetry about 0, we get
\begin{equation}\label{eq:cases}
    q(x) =
    \begin{cases}
        \half a\,(x-a)^{-2}, & \text{for } x<0\\
        \half a\,(x+a)^{-2}, & \text{for } x\geq0.
     \end{cases}
\end{equation}
Here we motivate the natural objective choice for the constant $a$ as 1. The only important transformation to be considered here is $\phi=1/\theta$. This, for example, would take variance to precision. For the prior in \eqref{eq_prior2} to be invariant, that is $p_a(\phi)=a/(a+\phi)^2$, we need to have $a=1$, since
\begin{eqnarray*}
p_a(\phi) = \frac{a}{(a+1/\phi)^2}\left|\frac{1}{\phi^2}\right| 
= \frac{a}{(a\phi + 1)^2},
\end{eqnarray*}
which yields $p_a(\theta)=a/(a+\theta)^2$ iff $a=1$. All the illustrations that follow have been made taking this choice for $a$.

\subsection{First examples}

The first simulation study was to make inference on a scale parameter;  specifcally the standard deviation of a normal density with mean $\mu=0$ and standard deviation $\sigma\in(0,\infty)$. We compare prior \eqref{eq_prior2} with Jeffreys prior, that is $\pi(\sigma)\propto1/\sigma$. We took 250 samples of size $n=100$, obtained the posterior distributions using standard MCMC methods (6000 iterations, with a burn--in of 1000 and a thinning of 10) and computed the following two indexes. 
The  root Mean Squared Error (MSE) divided by the true parameter value from the sample mean;
$$\mbox{MSE} = \frac{\sqrt{\mathbb{E}(\widehat{\sigma}-\sigma)^2}}{\sigma},$$
where $\widehat{\sigma}$ is the posterior mean, and the coverage of the 95\% posterior credible interval for $\sigma$. Table \ref{tab:tab1} shows the results for the MSE for $\sigma=\{0.25,0.50,1,2,5,10,20\}$, where we see little difference between the performance of the two priors.
However, the important point is that the score prior is proper, an important property.

\begin{table}[h!]
\centering
\begin{tabular}{|c|cc|}
\hline 
$\sigma$ & Jeffreys Prior & Score Prior \\ 
\hline 
0.25 & 0.0723 & 0.0720 \\ 
0.50 & 0.0721 & 0.0721 \\ 
1 & 0.0721 & 0.0716 \\ 
2 & 0.0719 & 0.0716 \\ 
5 & 0.0722 & 0.0720 \\ 
10 & 0.0723 & 0.0716 \\ 
20 & 0.0722 & 0.0718 \\ 
\hline 
\end{tabular}
\caption{Posterior MSE for $\sigma$ to compare Jeffreys prior against the prior based on scores for data generated by a normal density with mean $\mu=0$ and unknown variance. This has been obtained on 250 samples of size $n=100$ and with standard deviation $\sigma=\{0.25,0.50,1,2,5,10,20\}$.}
\label{tab:tab1}
\end{table}

The coverage of the 95\% posterior credible interval is shown in Table \ref{tab:COVnormal}, where we can also see very similar behaviour between the two priors; although both show an average slightly lower than the nominal 0.95.
\begin{table}[h!]
\centering
\begin{tabular}{|c|cc|}
\hline 
$\sigma$ & Jeffreys Prior & Score Prior \\ 
\hline 
0.25 & 0.91 & 0.91 \\ 
0.50 & 0.92 & 0.91 \\ 
1 & 0.91 & 0.91 \\ 
2 & 0.93 & 0.91 \\ 
5 & 0.93 & 0.91 \\ 
10 & 0.90 & 0.92 \\ 
20 & 0.90 & 0.93 \\ 
\hline 
\end{tabular}
\caption{Posterior coverage of the 95\% credible interval for $\sigma$ to compare Jeffreys prior against the prior based on scores for data generated by a normal density with mean $\mu=0$ and unknown variance. This has been obtained on 250 samples of size $n=100$ and with standard deviation $\sigma=\{0.25,0.50,1,2,5,10,20\}$.}
\label{tab:COVnormal}
\end{table}

To illustrate the frequentist properties of the prior in \eqref{eq:cases}, we have compared it to a flat prior, $\pi(\mu)\propto1$, in making inference for a location parameter of a log--normal density with unknown $\mu$ and known scale parameter $\sigma=1$. Similar to the previous case, we have drawn 250 samples of size $n=100$ and computed the MSE and coverage of the 95\% posterior credible interval. The values of $\mu$ considered were from the set $\{0,1,5,50,100\}$. Table \ref{tab:tab2} shows the MSE for the two priors, where we see that, apart for a small difference for $\mu=0$, the two priors appear to perform in a very similar fashion.

\begin{table}[h!]
\centering
\begin{tabular}{|c|cc|}
\hline 
$\mu$ & Jeffreys Prior & Score Prior \\ 
\hline 
0 & 0.0114 & 0.0007 \\ 
1 & 0.0086 & 0.0091 \\ 
5 & 0.0085 & 0.0085 \\ 
10 & 0.0085 & 0.0085 \\ 
50 & 0.0085 & 0.0087 \\ 
100 & 0.0085 & 0.0087 \\
\hline 
\end{tabular}
\caption{Posterior MSE for $\mu$ to compare Jeffreys prior against the prior based on scores for data generated from a log-normal density with unknown location parameter $\mu$ and known scale parameter $\sigma=1$. This has been obtained on 250 samples of size $n=100$ and with standard deviation $\mu=\{0,1,1,5,50,100\}$.}
\label{tab:tab2}
\end{table}

The coverage of the 95\% posterior credible interval for $\mu$ is shown in Table~\ref{tab:COVlognormal}, where we can see a very similar behaviour for the two priors, with an exception for $\mu=0$, although the two coverage level are perfectly acceptable.

\begin{table}[h!]
\centering
\begin{tabular}{|c|cc|}
\hline 
$\mu$ & Jeffreys Prior & Score Prior \\ 
\hline 
0 & 0.92 & 1.00 \\ 
1 & 0.97 & 0.98 \\ 
5 & 0.97 & 0.98 \\ 
10 & 0.97 & 0.98 \\ 
50 & 0.97 & 0.98 \\ 
100 & 0.97 & 0.98 \\
\hline 
\end{tabular}
\caption{Posterior coverage of the 95\% credible interval for $\mu$ to compare Jeffreys prior against the prior based on scores for data generated from a log-normal density with unknown location parameter $\mu$ and known scale parameter $\sigma=1$. This has been obtained on 250 samples of size $n=100$ and with standard deviation $\mu=\{0,1,1,5,50,100\}$.}
\label{tab:COVlognormal}
\end{table}

The general conclusion for the two experiments above, is that the prior obtained via $\alpha(u)=1/u^2$ exhibits tails which are sufficiently heavy to generate optimal frequentist performance even for large parameter values. These properties are comparable to those obtained by Jeffreys prior, which is well--known for being the objective prior yielding good frequentist properties of the posterior. The advantage with our prior is that it is always proper.

\section{Mixture models}\label{sc_mixtures}
A major area of challenge for objective priors is finite mixture models, where observations are assumed to be generated by the following model;
\begin{equation}\label{eq_mixmodel}
f(y) =\sum_{l=1}^k \omega_l\,\,f_l(y\mid \theta_l), \qquad \sum_{l=1}^k\omega_l=1,
\end{equation}
for densities $(f_l(\cdot|\theta_l))$, where $\theta_l$ is vector of parameters characterising the densities. Given the ill-defined nature of the model in \eqref{eq_mixmodel}, \cite{GrazianRobert2018}, the use of improper priors for the parameters is not acceptable. In particular, if we consider densities $f_l$ to be location-scale distributions, \cite{GrazianRobert2018} show that, under certain circumstances, Jeffreys priors cannot be used, due to their improperness. For example, if all parameters are unknown (i.e. weights, location and scale parameters), then Jeffreys prior yields improper posteriors. Even in more restrictive circumstances the use of improper priors if troublesome; if we consider only the location parameters unknown, then Jeffreys prior yields improper posteriors for $k>2$.  The above represents severe limitations in Bayesian analysis. Therefore, the objective priors proposed in this paper represent a clear solution to the above problem, avoiding reparametrisation or addition of extra parameters, as proposed, for example in \cite{GrazianRobert2018}, the latter resulting in an increased uncertainty. 


\subsection{Single sample}
In this first simulation study, we illustrate the performance of the proposed prior on the following three-component mixture model, from which we have drawn a sample of size $n=200$,
\begin{equation}\label{example1}
0.25\,\mbox{N}(0,1.2^2) + 0.65\,\mbox{N}(-10,1) + 0.10\,\mbox{N}(7,0.8^2).
\end{equation}
In terms of prior distributions, we have assumed a symmetric Dirichlet prior with concentration parameters equal to one, and for the means and standard deviations of the Gaussian components the proposed prior on $(-\infty,\infty)$ and on $(0,\infty)$, respectively. We have also assumed independent information, so the priors for the parameters of the component have the following form,
$$\pi(\pmb\mu,\pmb\sigma) = \prod_{l=1}^3\pi_l(\mu_l)\times\prod_{l=1}^3\pi(\sigma_l),$$
where $\pmb\mu=(\mu_1,\mu_2,\mu_3)$ and $\pmb\sigma=(\sigma_1,\sigma_2,\sigma_3).$ The histogram of the sample, together with the true density, is shown in Fig.~\ref{fig:example1_hist}.

\begin{figure}[h]
\centering
\includegraphics[width=14cm,height=6cm]{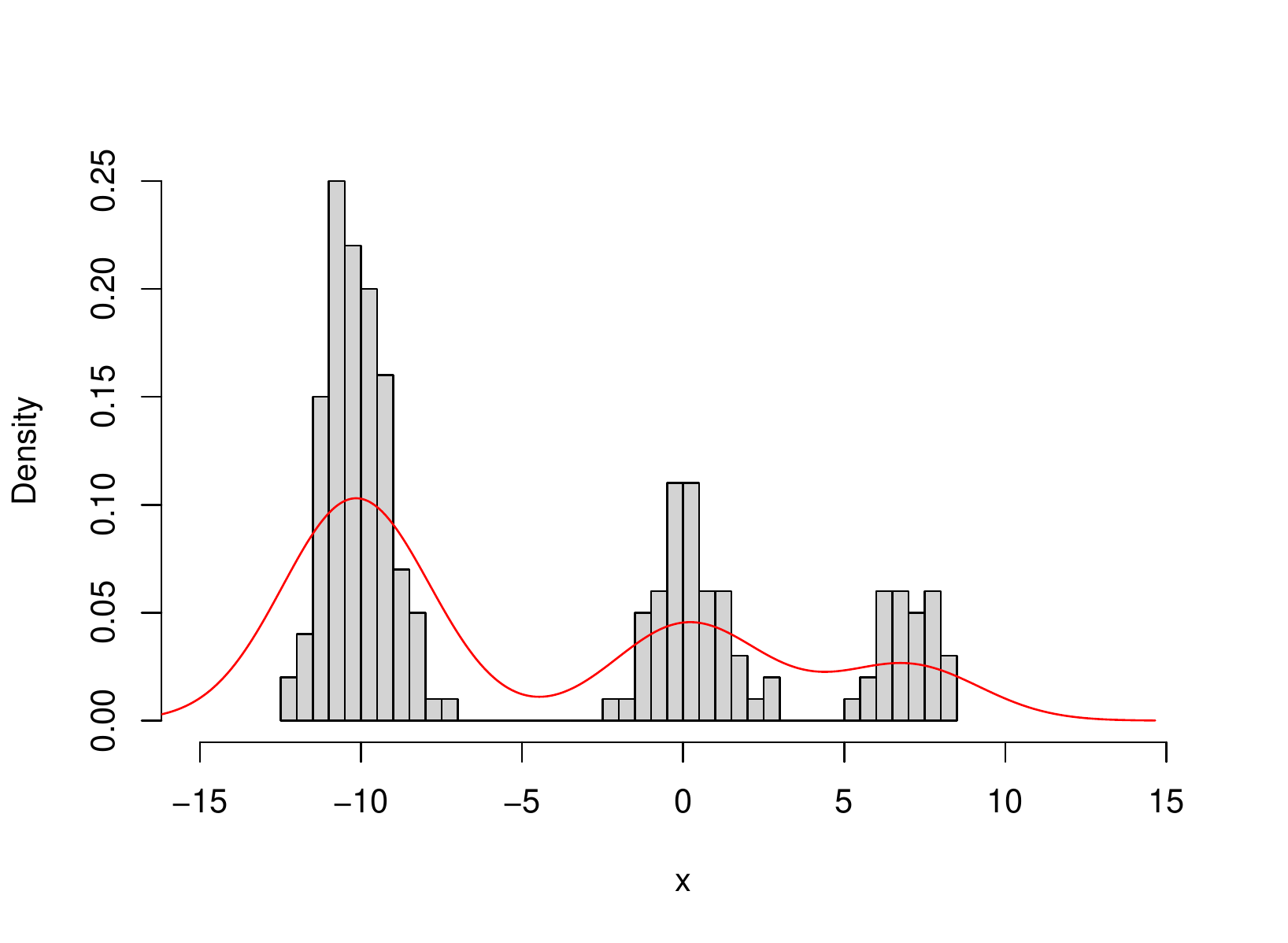}
\caption{Histogram of the sample of size $n=200$ from model \eqref{example1}, and density (red line) of the true model.}
\label{fig:example1_hist}
\end{figure}

The analysis uses a Metropolis--within--Gibbs algorithm with a total of 60000 iterations, a burn--in of 10000 and a thinning of 100. We note that the true values are within the posterior credible intervals.

\begin{table}[htbp]
  \centering
    \begin{tabular}{c|ccc}
    \hline
    Component & $\omega$ & $\mu$  & $\sigma$ \\
    \hline
    \multirow{2}[0]{*}{1} & 0.26  & -10.2  & 1.1 \\
          & (0.21, 0.33) & (-10.5, -9.9) & (0.9, 1.4) \\
    \multirow{2}[0]{*}{2} & 0.66  & 0.0 & 1.3 \\
          & (0.58,0.72) & (-0.1, 0.2) & (1.1, 1.5) \\
    \multirow{2}[0]{*}{3} & 0.08  & 6.7  & 0.9 \\
          & (0.04, 0.12) & (6.2, 7.2) & (0.6, 1.4) \\
          \hline
    \end{tabular}
    \caption{Posterior means and 95\% credible intervals (in brackets) for a sample of size $n=200$ from model \eqref{example1}.}
  \label{tab:singlesample_100}
\end{table}

\subsection{Repeated sampling}
To have a more thorough understanding of the proposed prior implementation, we have performed some experiment on repeated sampling, taking into consideration different scenarios, which include different sample sizes and model structure. We have limited the analysis to mixture of normal densities, but it is obvious that, due to the properness of the prior, its implementation can be extended to any family of densities for the mixture components. We computed the posterior distribution for $M=20$ replications of sample of size $n=(50, 100, 200)$ of mixture models with number of components $k=(3,4,5)$. For this illustrations, we reported the results for the means and the standard deviations of each components, as they are estimated using the proposed prior. The models used are as follows:
\begin{eqnarray}
\frac{1}{3}\mbox{N}(-10,1) &+& \frac{1}{3}\mbox{N}(0,0.8^2) + \frac{1}{3}\mbox{N}(7,1.2^2) \nonumber \\
\frac{1}{4}\mbox{N}(-10,1) &+&\frac{1}{4}\mbox{N}(-3,0.9) +\frac{1}{4}\mbox{N}(0,0.8) +\frac{1}{4}\mbox{N}(7,1.2) \nonumber \\
\frac{1}{5}\mbox{N}(-10,1) &+& \frac{1}{5}\mbox{N}(-3,0.9) +\frac{1}{5}\mbox{N}(0,0.8) +\frac{1}{5}\mbox{N}(3,1) +\frac{1}{5}\mbox{N}(7,1.2). \nonumber
\end{eqnarray}
Note that we have not chosen variable weights as these are not associated with a proposed prior.

Fig.~\ref{fig:means} shows the boxplots of the posterior means for the $\pmb\mu$ of the mixture components, while Fig.~\ref{fig:sds} shows the boxplots of the posterior means for the standard deviations $\pmb\sigma$. As one would expect, the larger the sample size the less variability in the repeated estimates, for the same number of components. Keeping the sample size fixed, we notice an increase in the variability of the estimates as the number of components grows, which is also an expected result.

\begin{figure}[htb]
    \centering 
\begin{subfigure}{0.25\textwidth}
  \includegraphics[width=\linewidth]{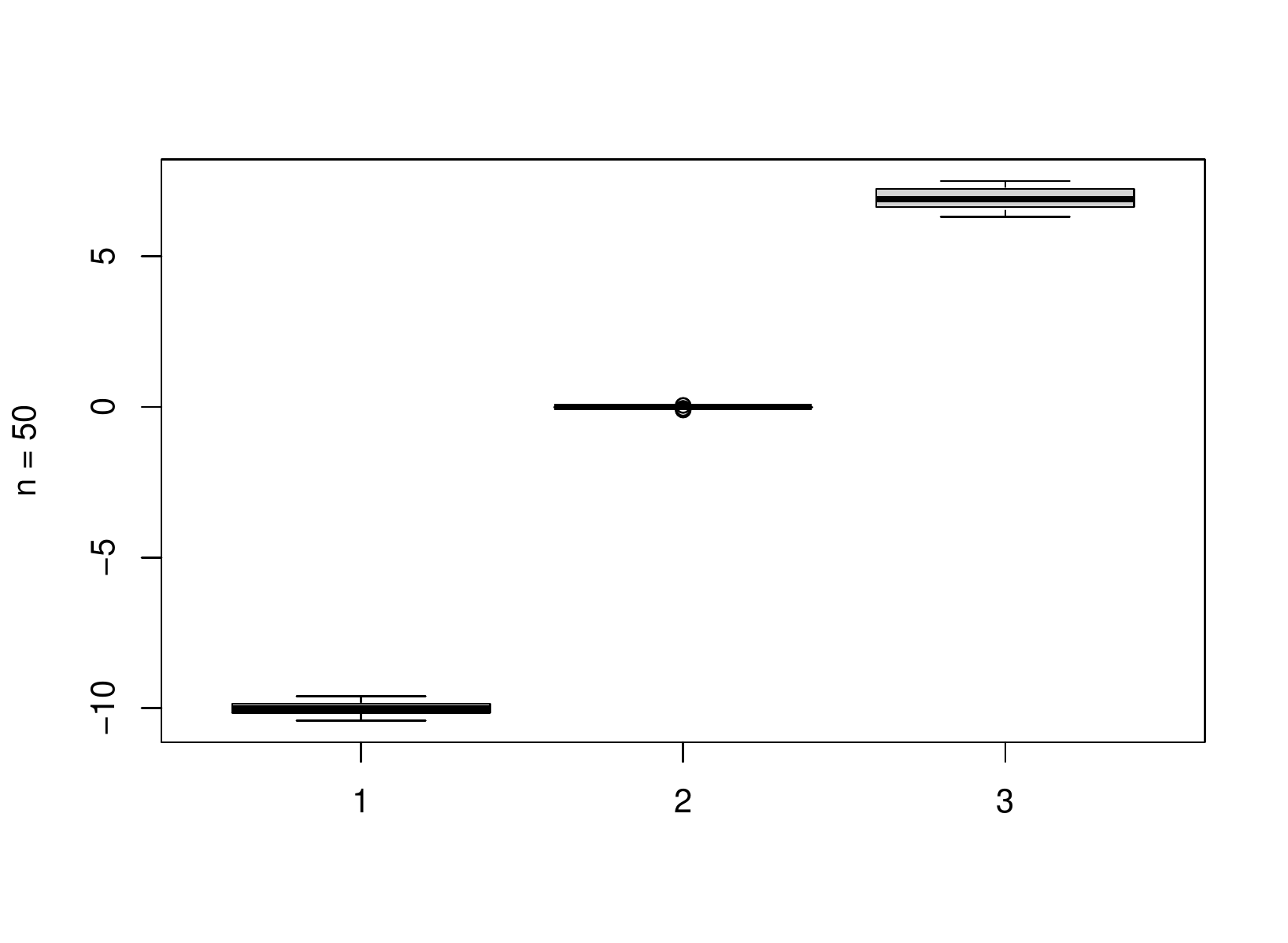}
  \label{fig:1a}
\end{subfigure}\hfil 
\begin{subfigure}{0.25\textwidth}
  \includegraphics[width=\linewidth]{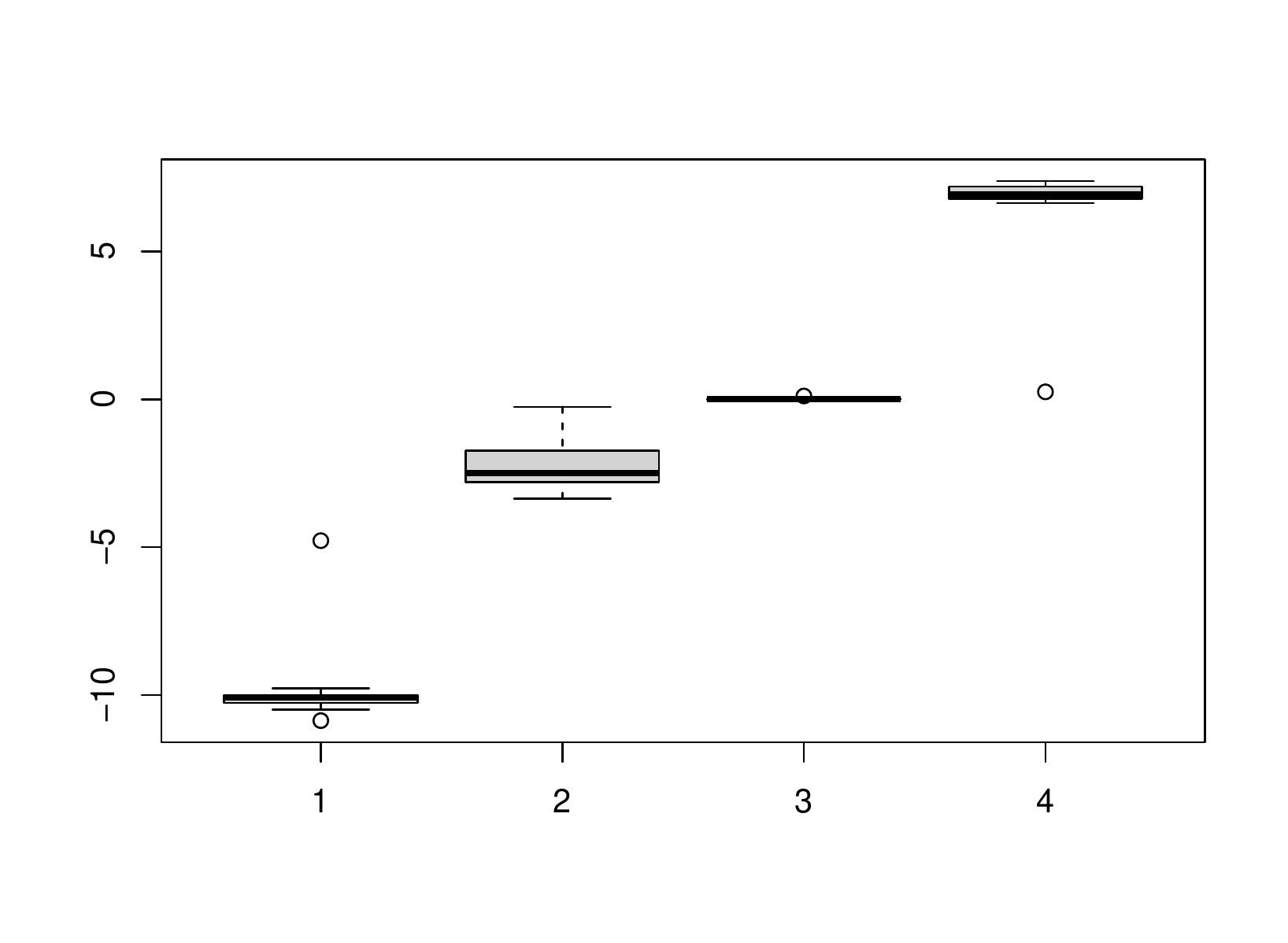}
  \label{fig:2a}
\end{subfigure}\hfil 
\begin{subfigure}{0.25\textwidth}
  \includegraphics[width=\linewidth]{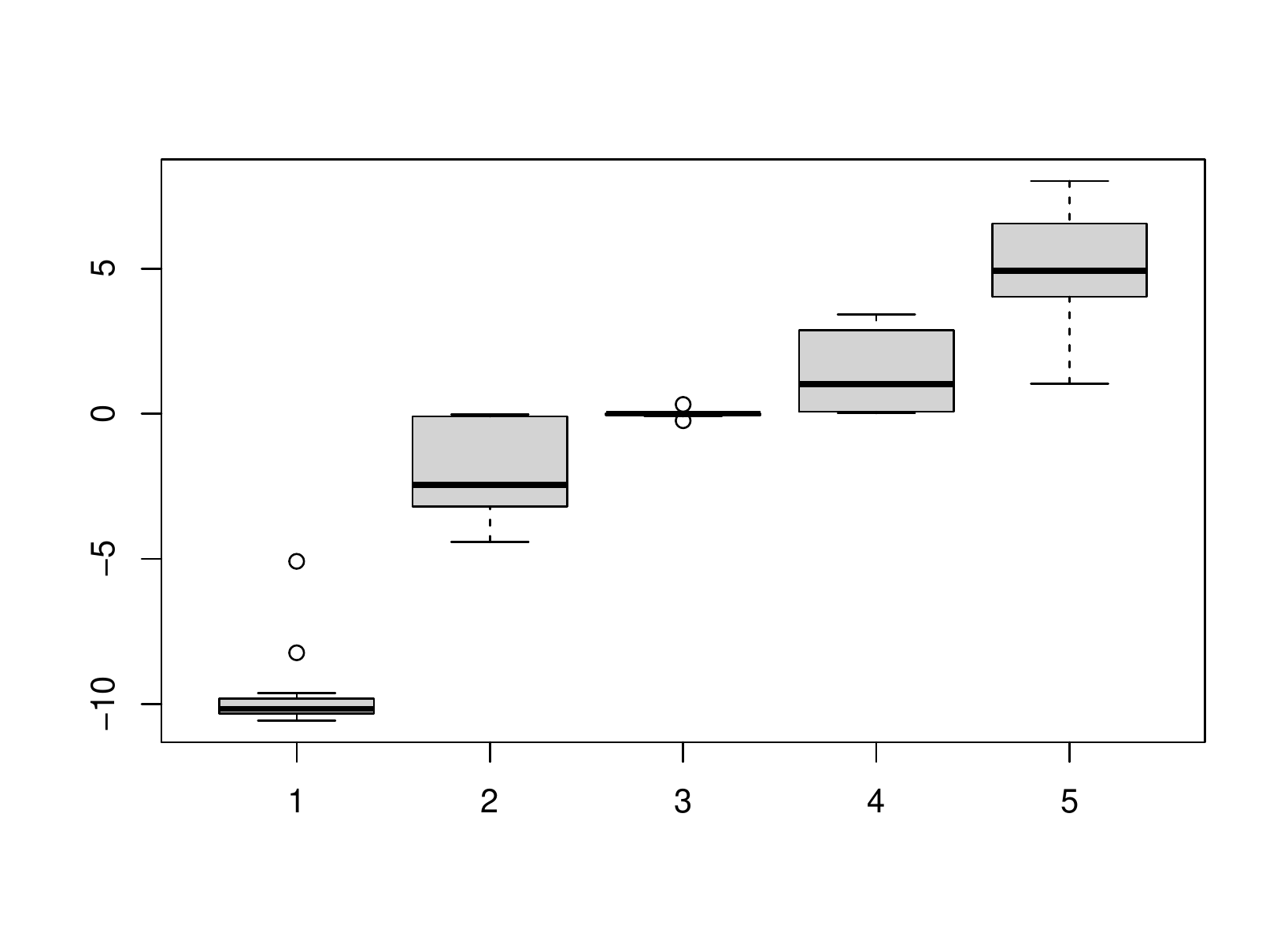}
  \label{fig:3a}
\end{subfigure}

\medskip
\begin{subfigure}{0.25\textwidth}
  \includegraphics[width=\linewidth]{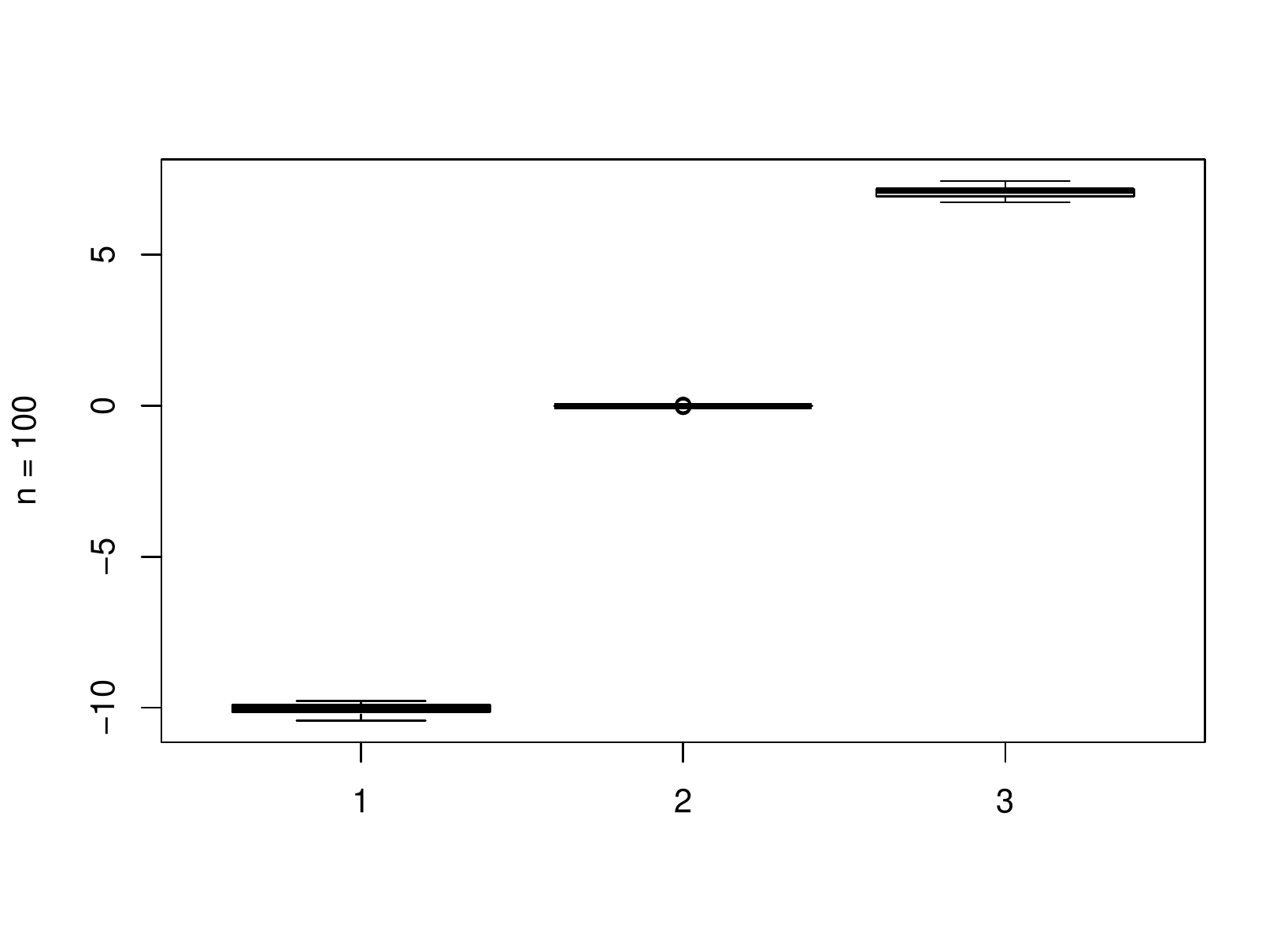}
  \label{fig:4a}
\end{subfigure}\hfil 
\begin{subfigure}{0.25\textwidth}
  \includegraphics[width=\linewidth]{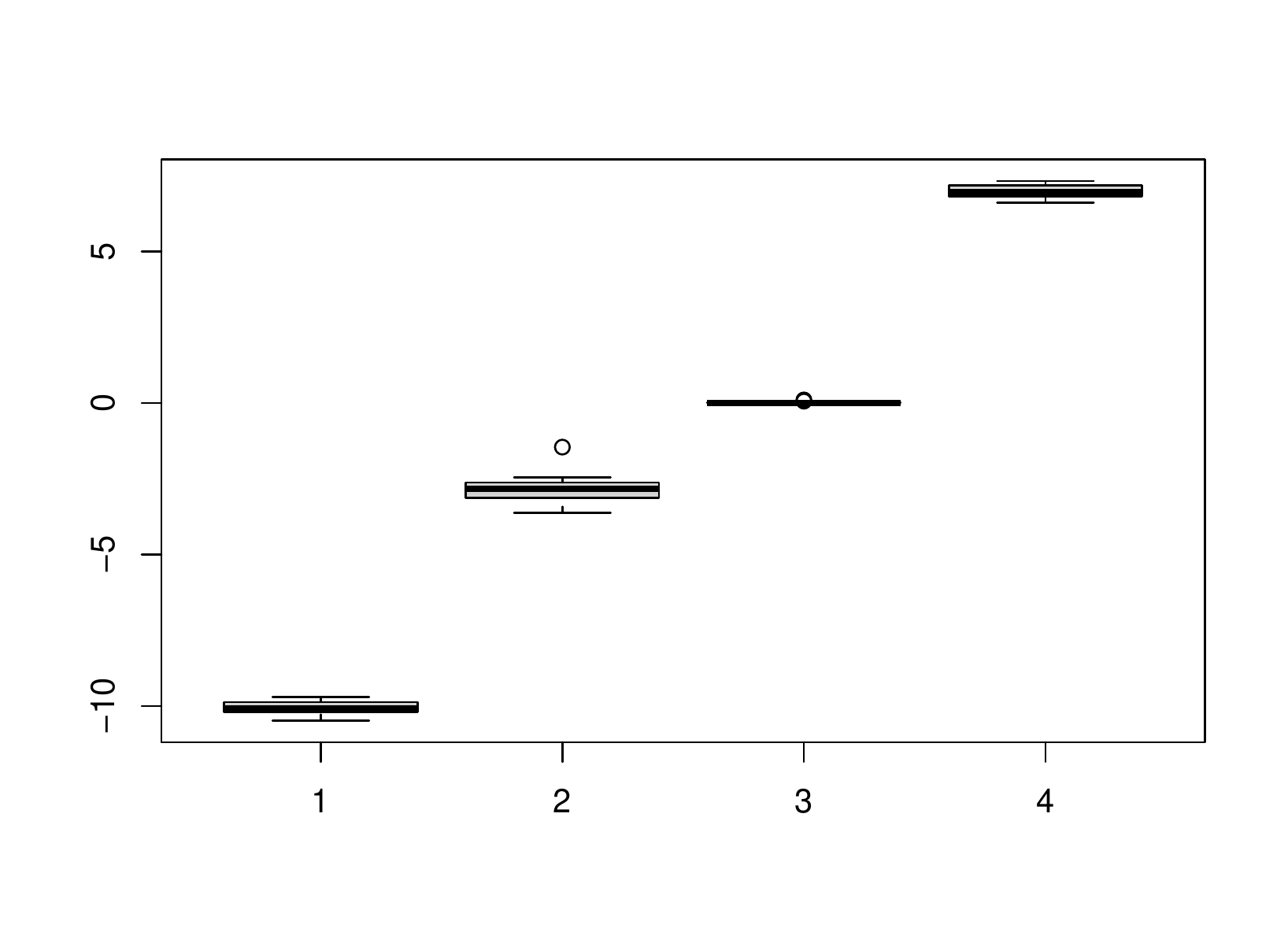}
  \label{fig:5a}
\end{subfigure}\hfil 
\begin{subfigure}{0.25\textwidth}
  \includegraphics[width=\linewidth]{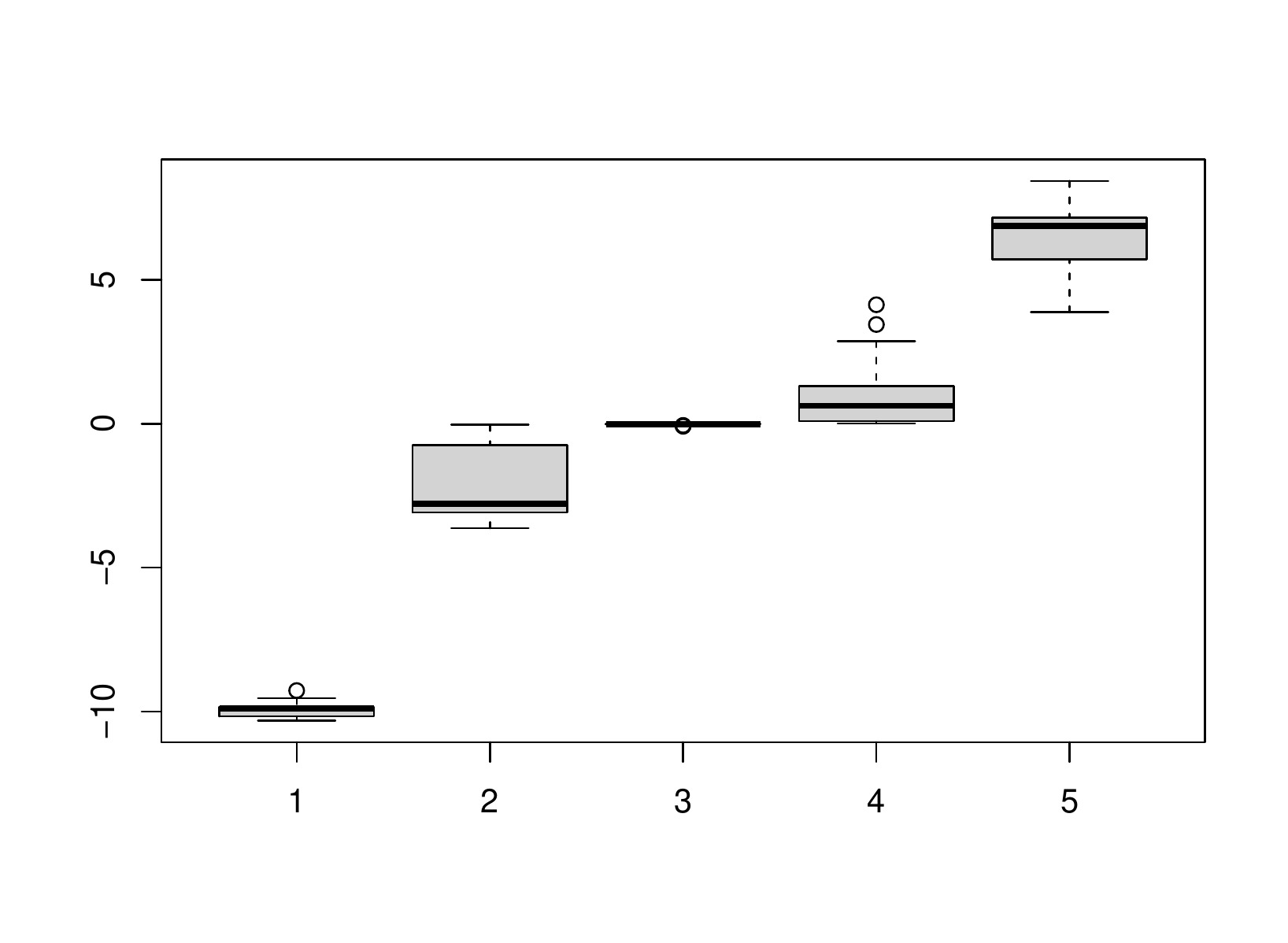}
  \label{fig:6a}
\end{subfigure}

\medskip
\begin{subfigure}{0.25\textwidth}
  \includegraphics[width=\linewidth]{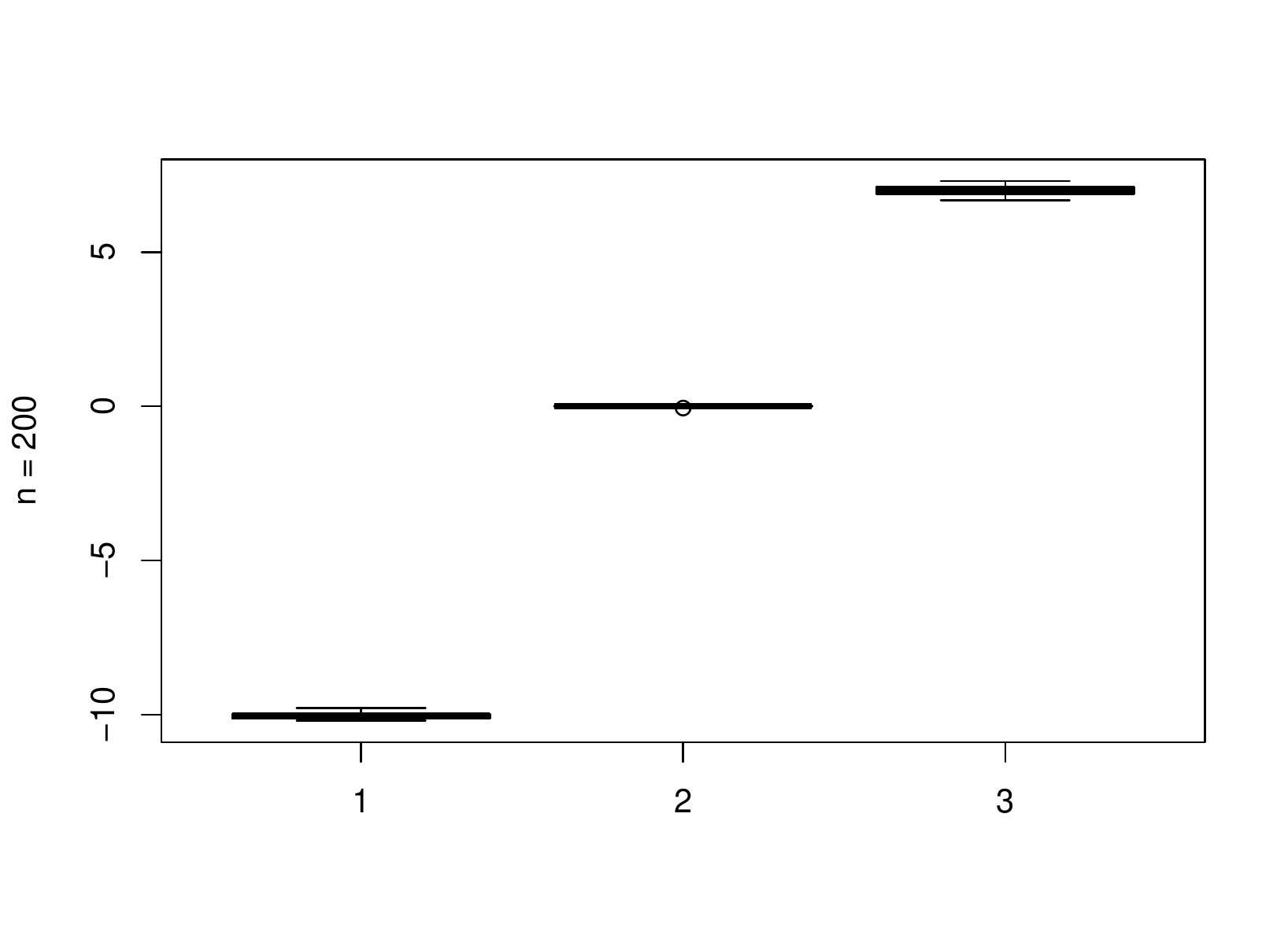}
  \label{fig:7a}
\end{subfigure}\hfil 
\begin{subfigure}{0.25\textwidth}
  \includegraphics[width=\linewidth]{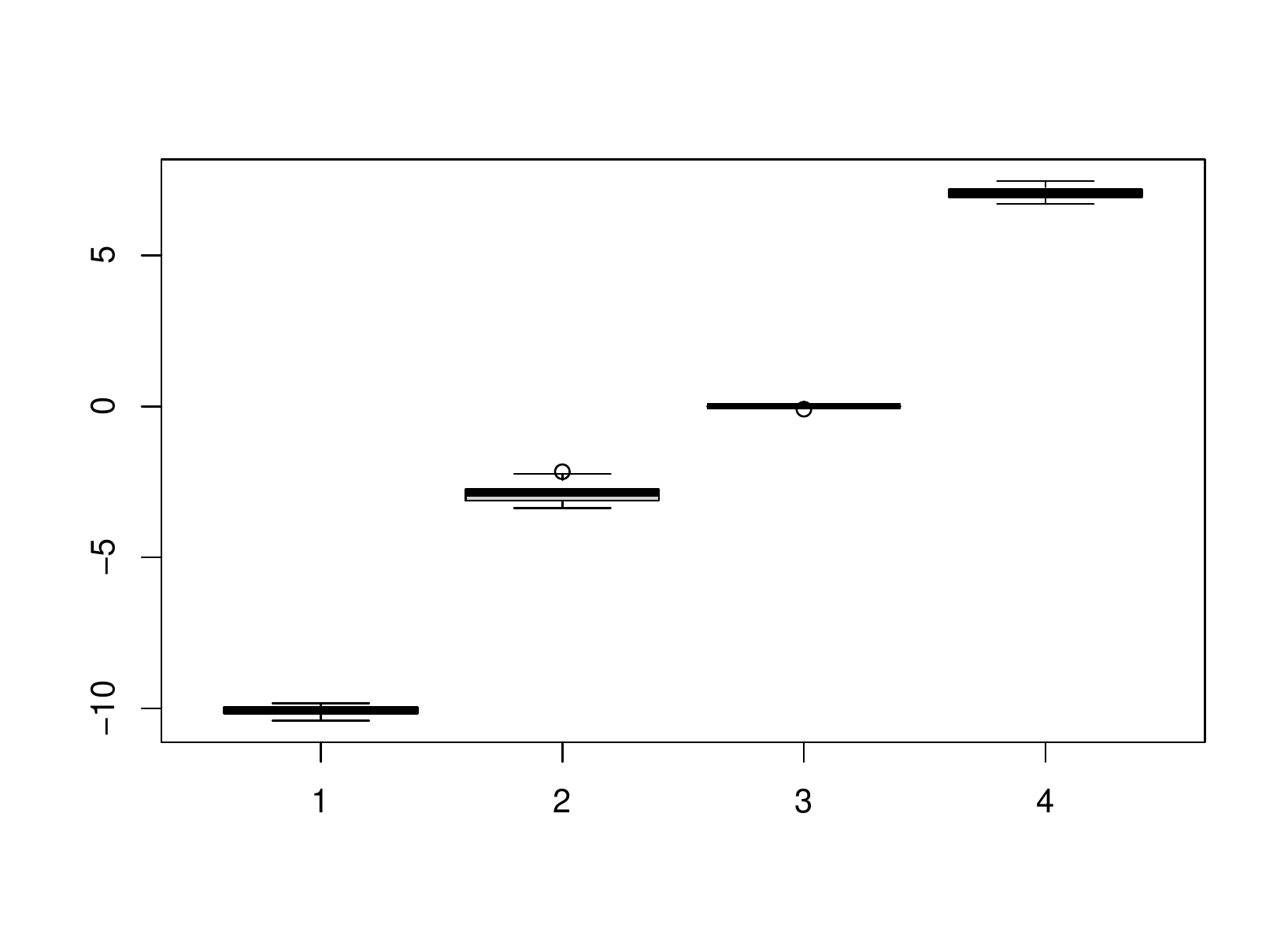}
  \label{fig:8a}
\end{subfigure}\hfil 
\begin{subfigure}{0.25\textwidth}
  \includegraphics[width=\linewidth]{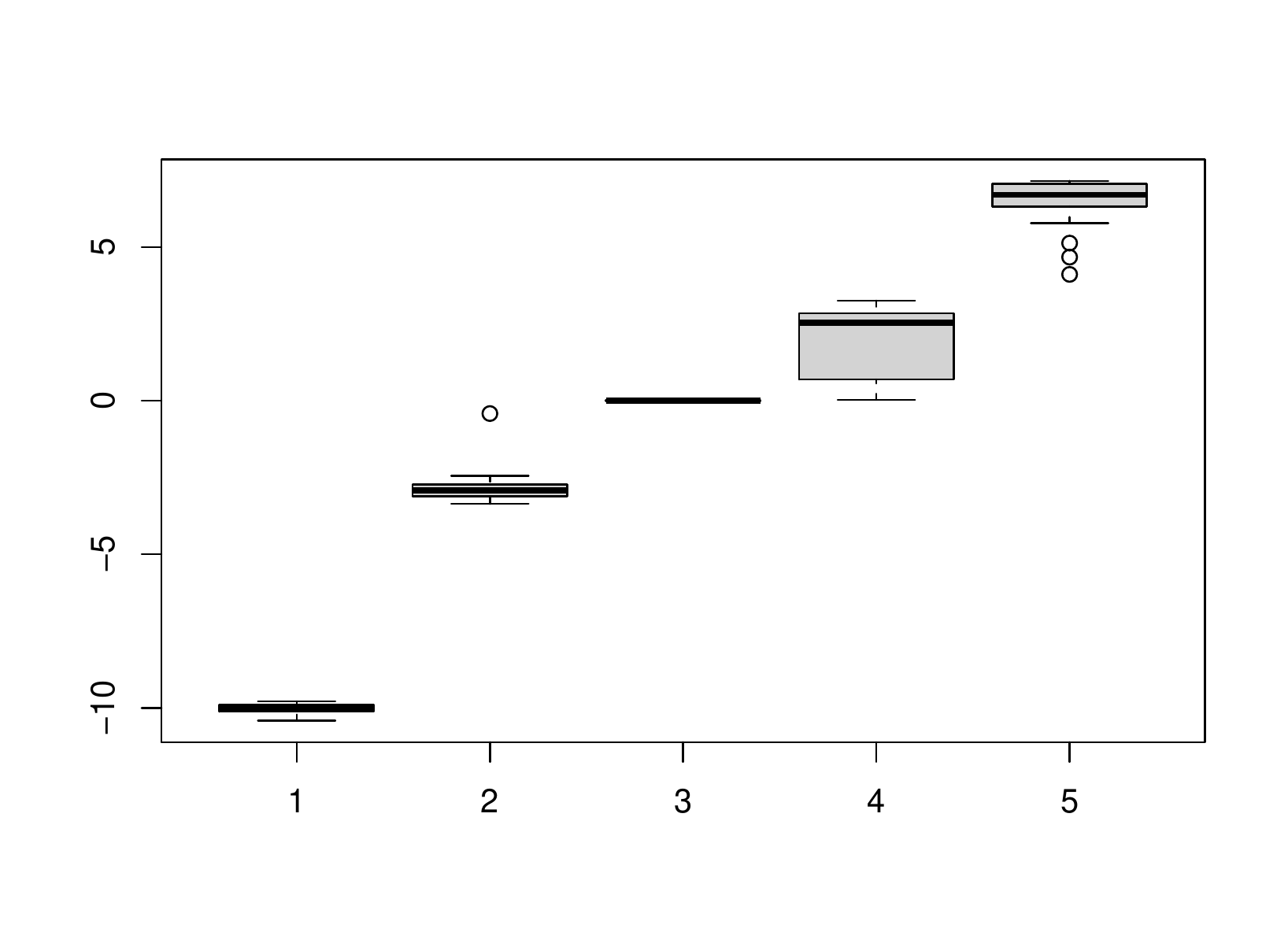}
  \label{fig:9a}
\end{subfigure}
\caption{Boxplots of posterior means of the means $\pmb\mu$, using the proposed prior, for samples of size $n=(50,100,200)$, plotted by row, and number of components $k=(3,4,5)$, plotted by column.}
\label{fig:means}
\end{figure}

\begin{figure}[htb]
    \centering 
\begin{subfigure}{0.25\textwidth}
  \includegraphics[width=\linewidth]{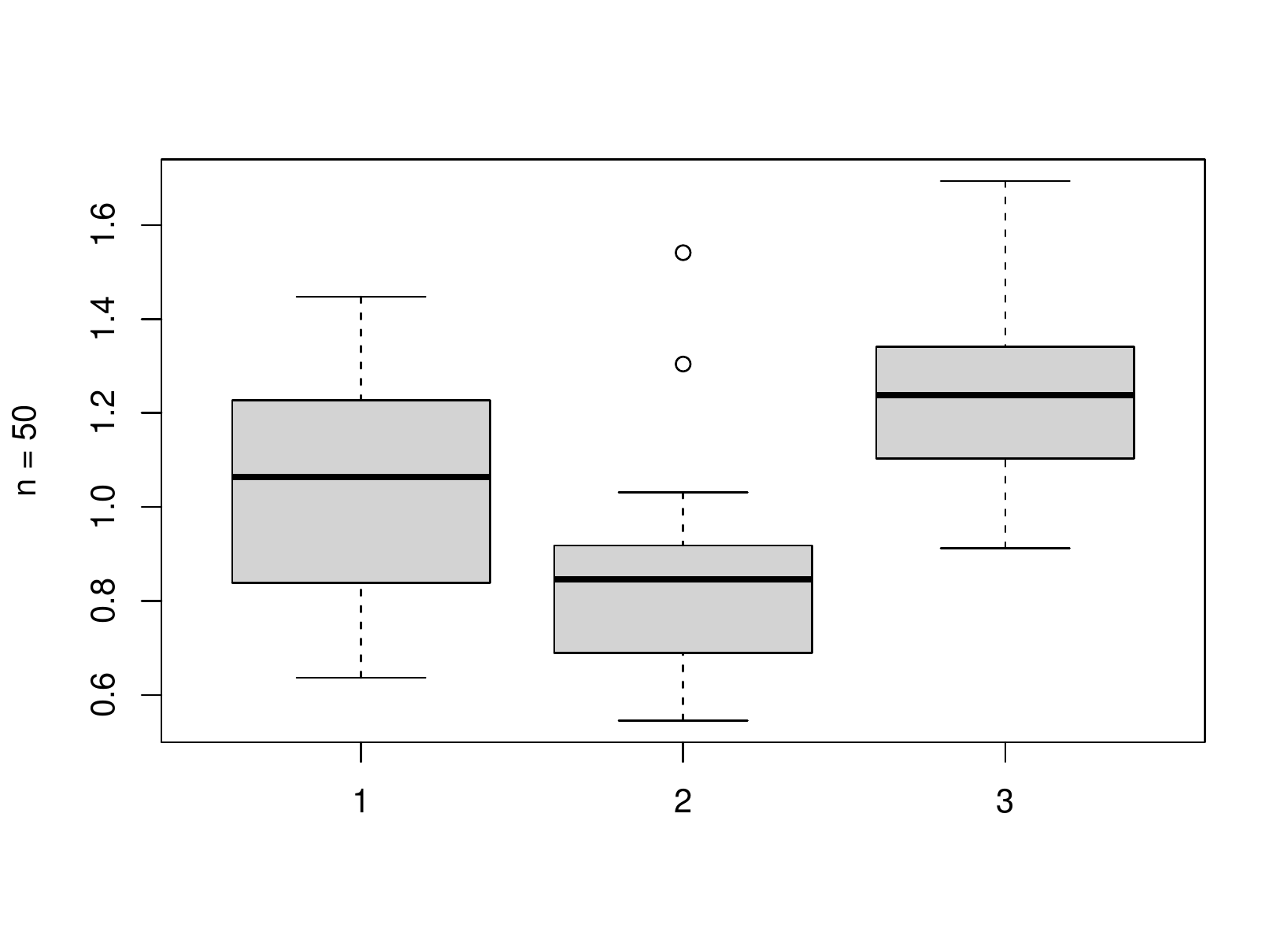}
  \label{fig:1b}
\end{subfigure}\hfil 
\begin{subfigure}{0.25\textwidth}
  \includegraphics[width=\linewidth]{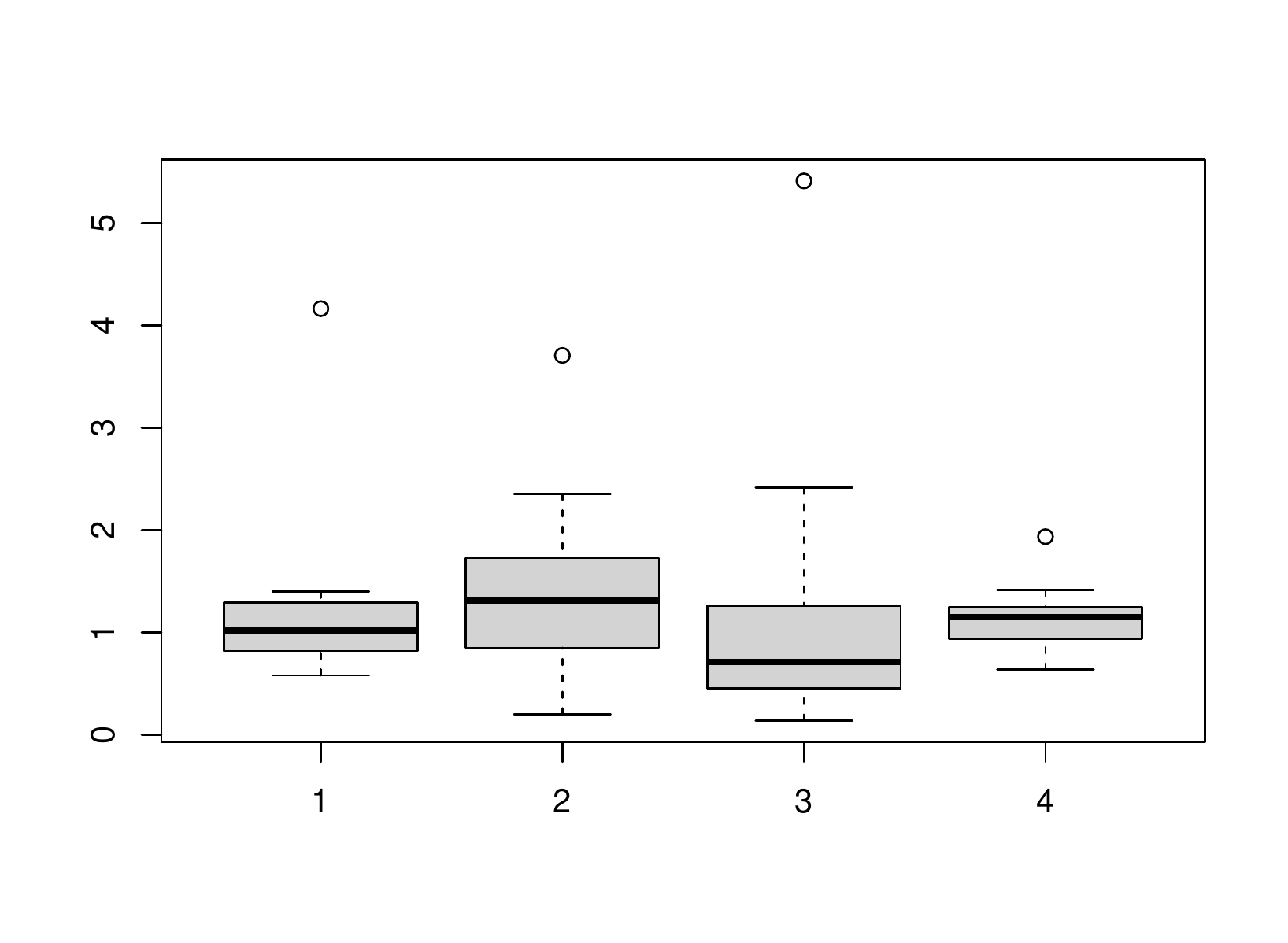}
  \label{fig:2b}
\end{subfigure}\hfil 
\begin{subfigure}{0.25\textwidth}
  \includegraphics[width=\linewidth]{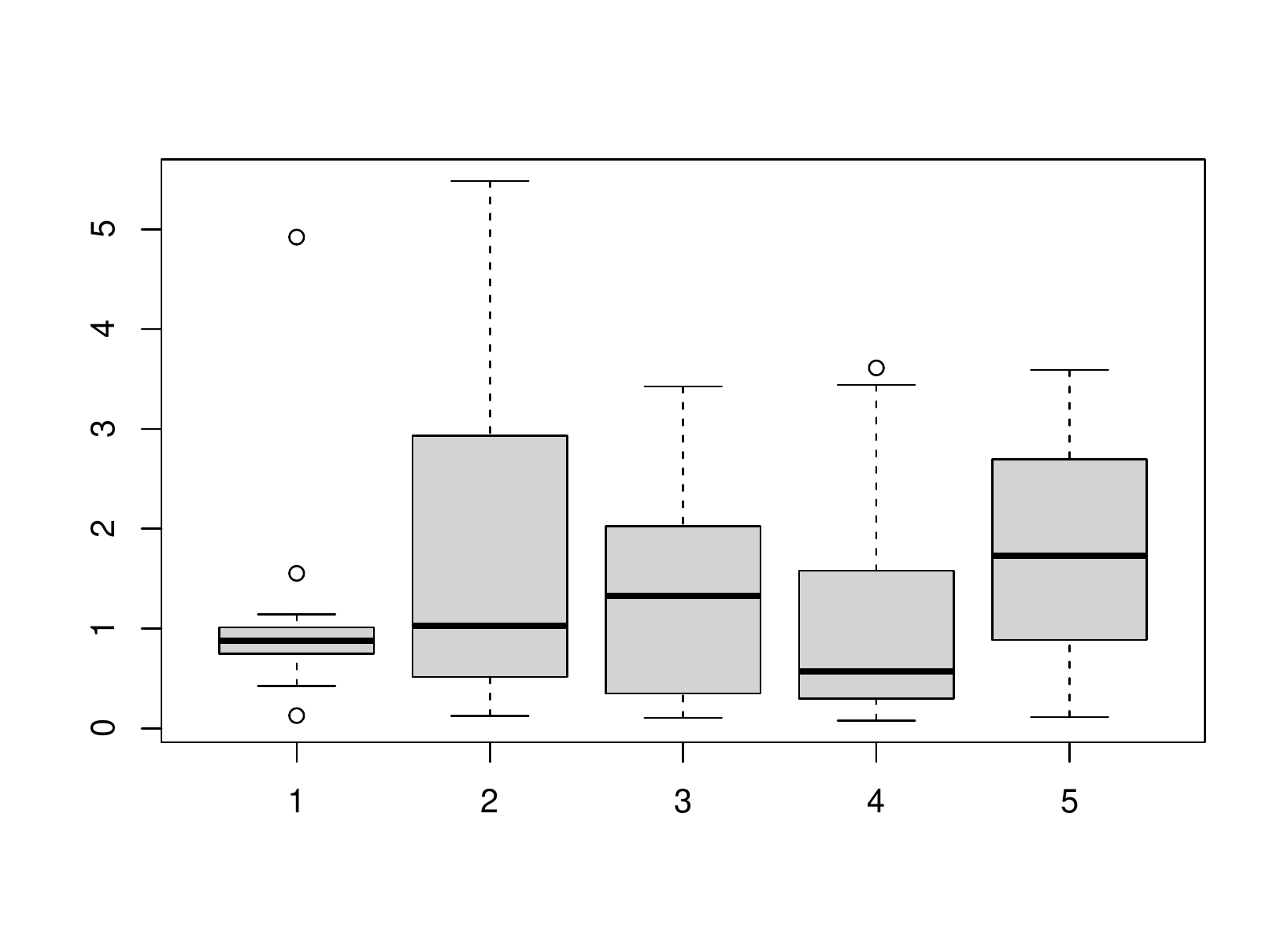}
  \label{fig:3b}
\end{subfigure}

\medskip
\begin{subfigure}{0.25\textwidth}
  \includegraphics[width=\linewidth]{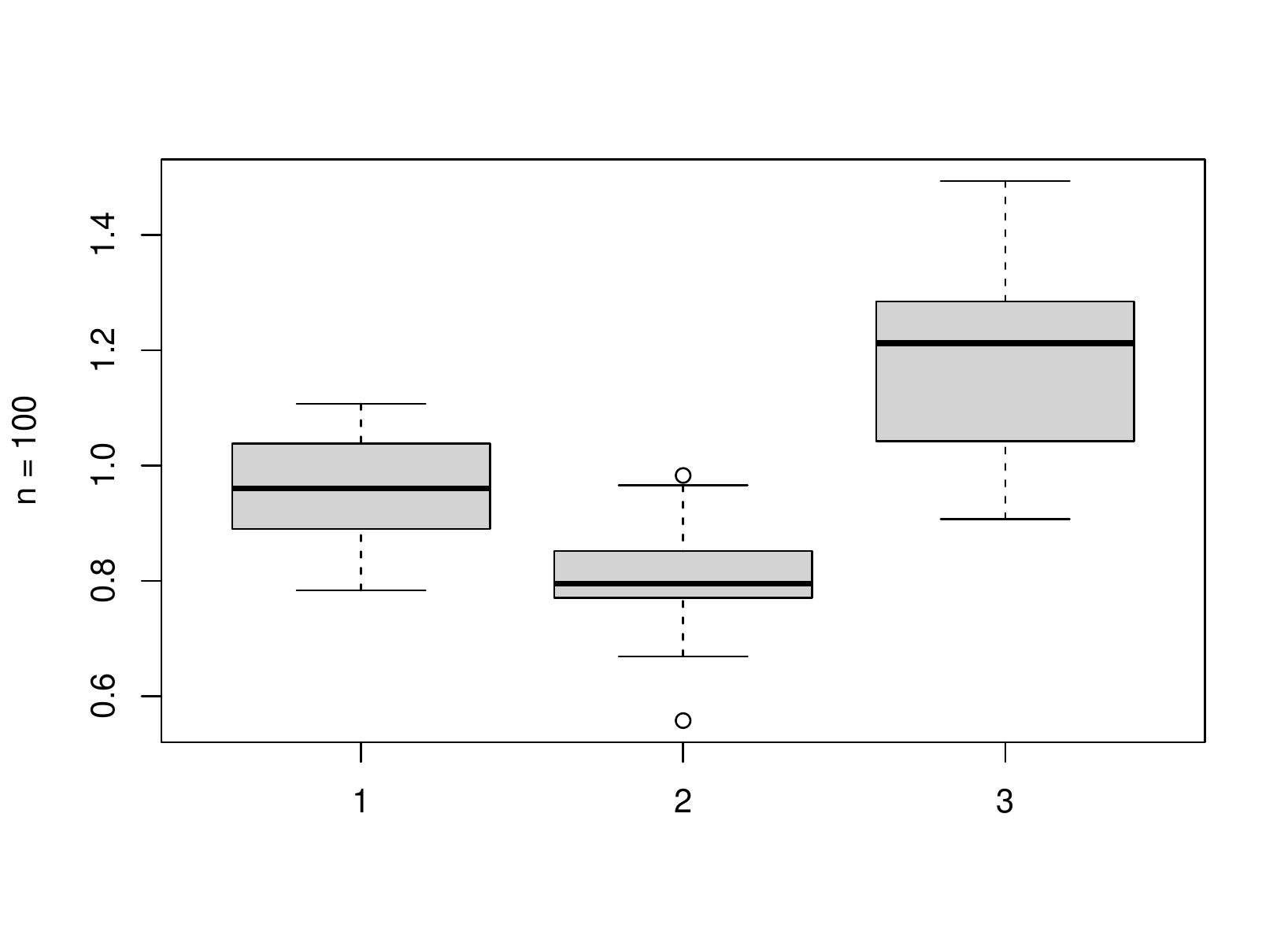}
  \label{fig:4b}
\end{subfigure}\hfil 
\begin{subfigure}{0.25\textwidth}
  \includegraphics[width=\linewidth]{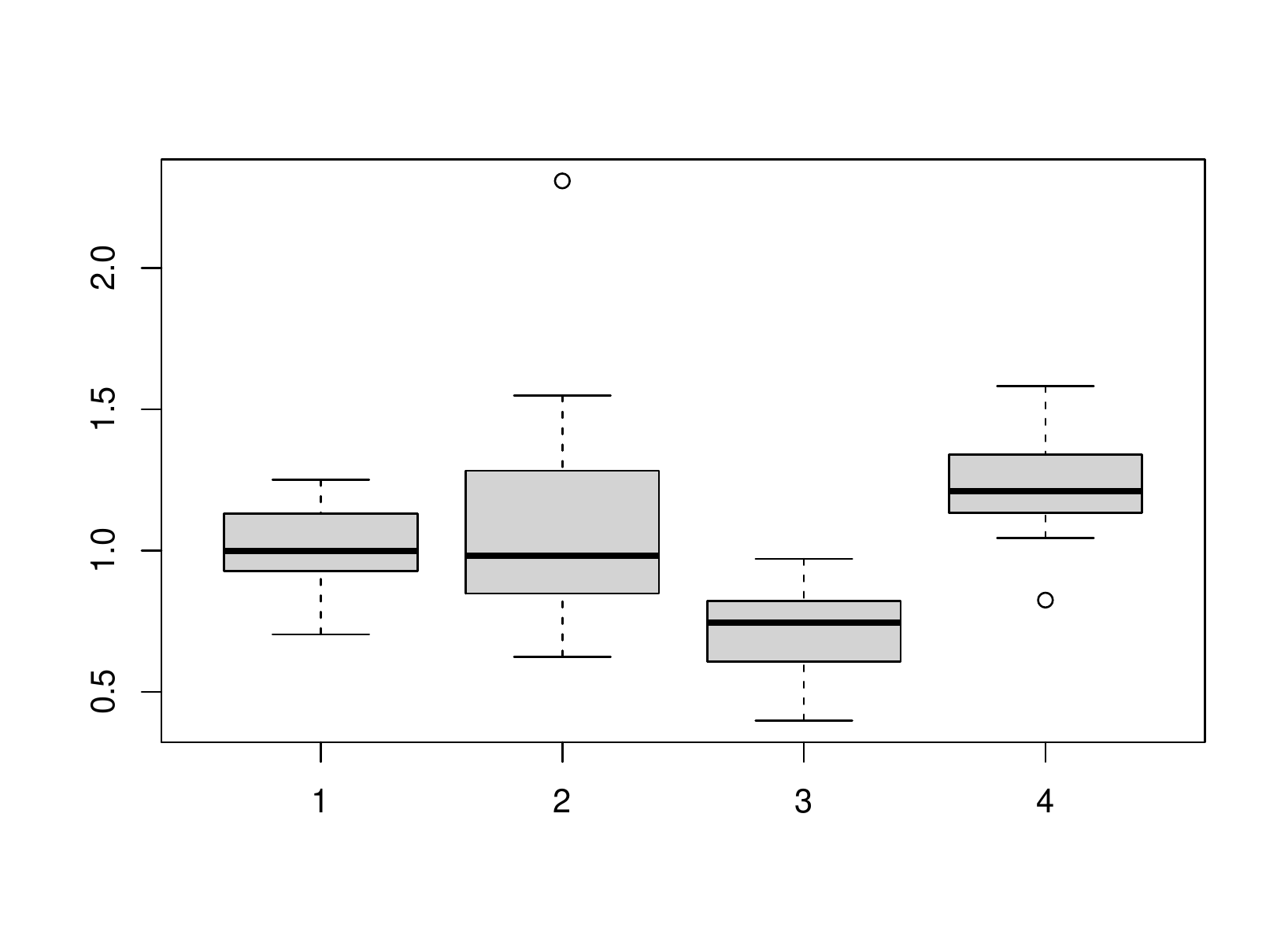}
  \label{fig:5b}
\end{subfigure}\hfil 
\begin{subfigure}{0.25\textwidth}
  \includegraphics[width=\linewidth]{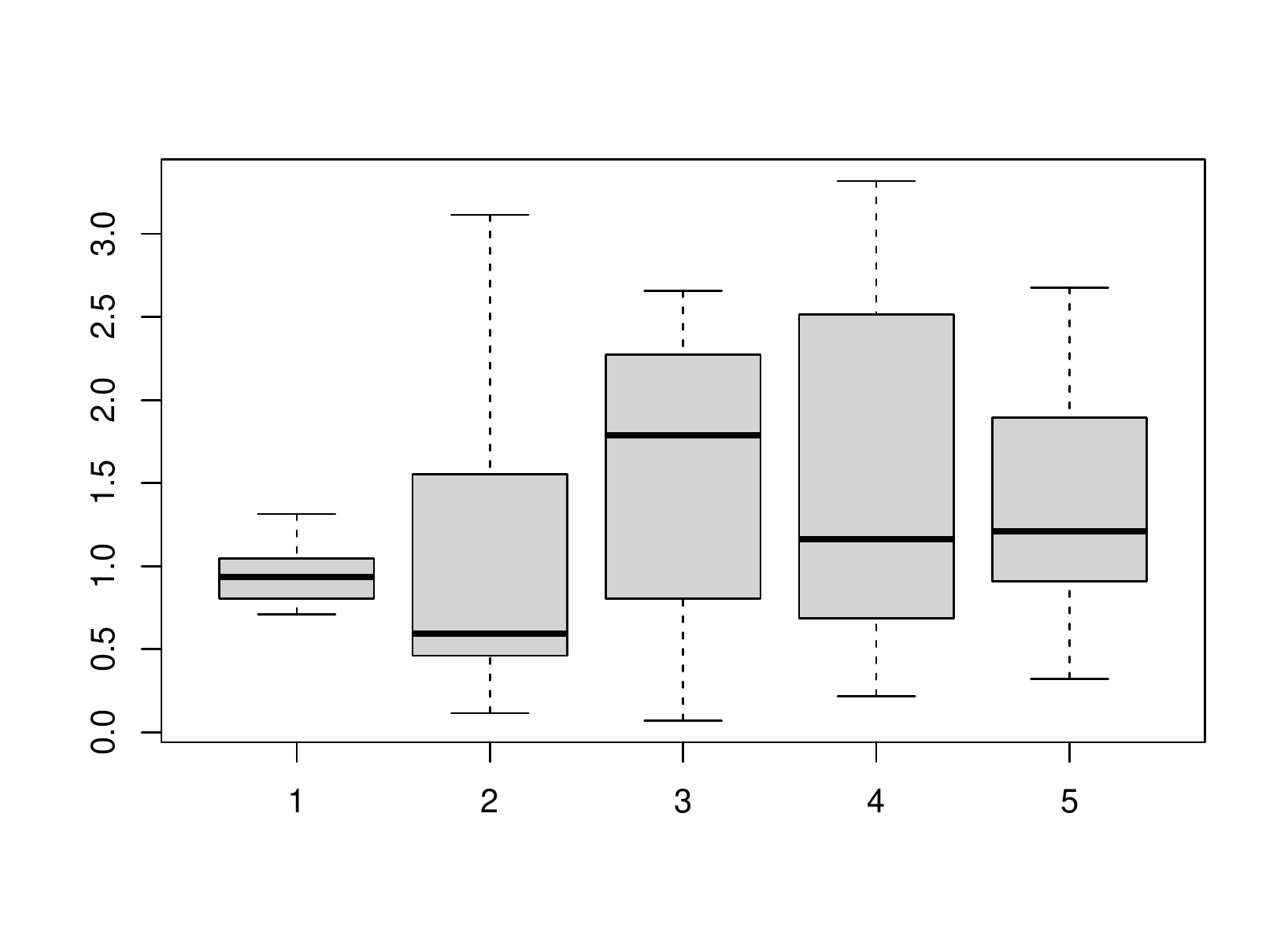}
  \label{fig:6b}
\end{subfigure}

\medskip
\begin{subfigure}{0.25\textwidth}
  \includegraphics[width=\linewidth]{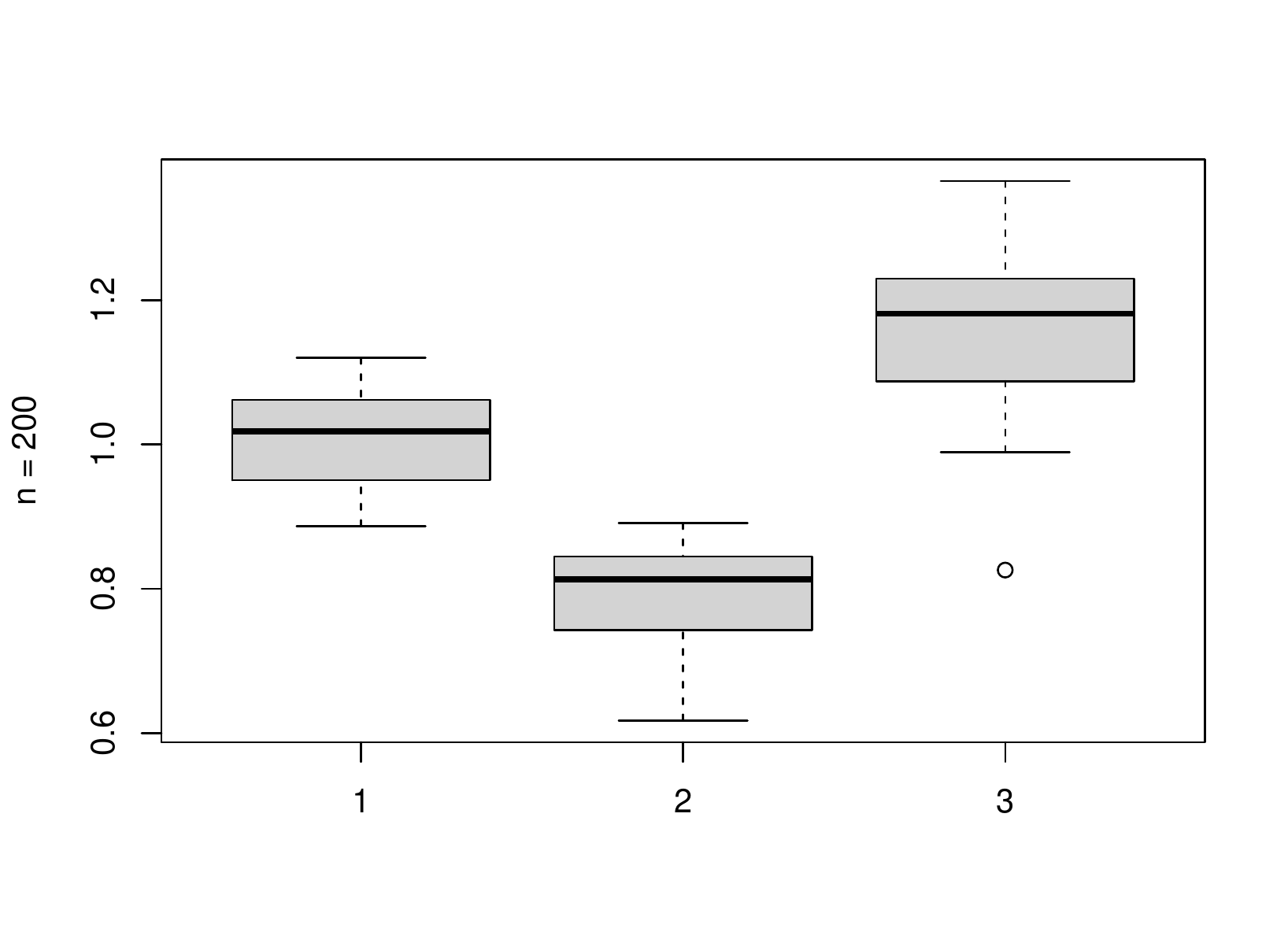}
  \label{fig:7b}
\end{subfigure}\hfil 
\begin{subfigure}{0.25\textwidth}
  \includegraphics[width=\linewidth]{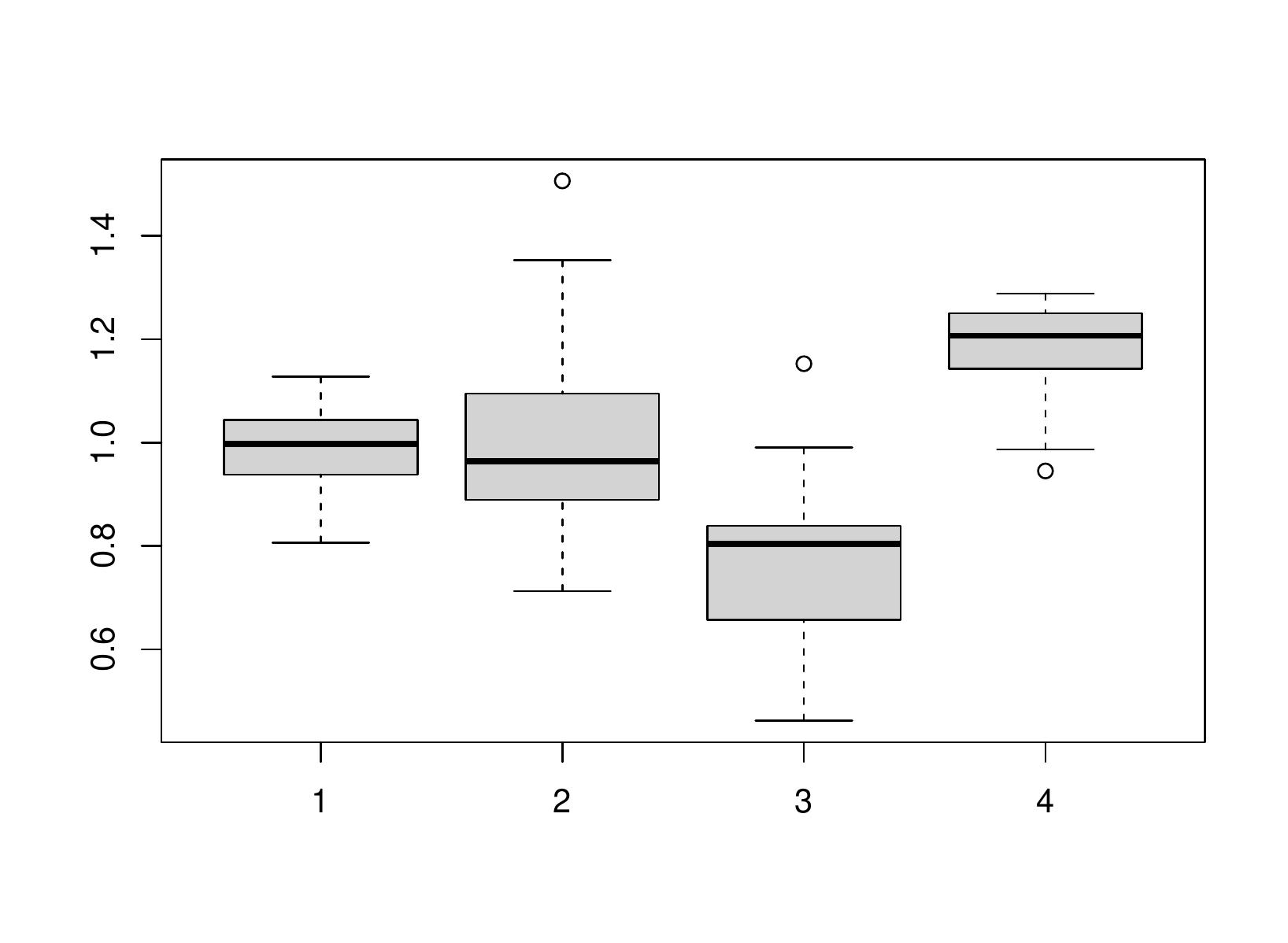}
  \label{fig:8b}
\end{subfigure}\hfil 
\begin{subfigure}{0.25\textwidth}
  \includegraphics[width=\linewidth]{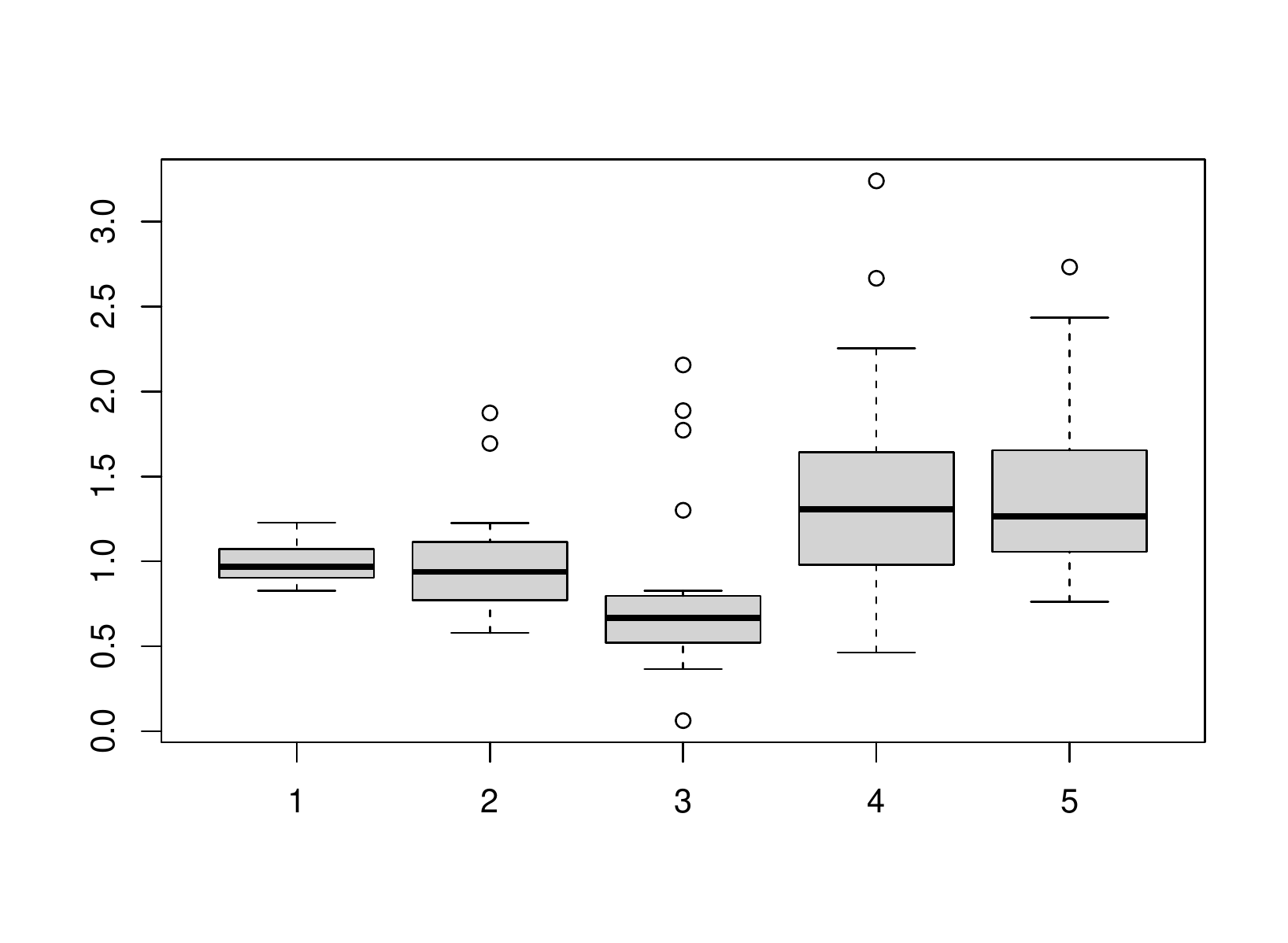}
  \label{fig:9b}
\end{subfigure}
\caption{Boxplots of posterior standard deviations of the means $\pmb\sigma$, using the proposed prior, for samples of size $n=(50,100,200)$, plotted by row, and number of components $k=(3,4,5)$, plotted by column.}
\label{fig:sds}
\end{figure}

\subsection{Real data analysis}
In this section we analyse the well--known galaxy data set, which contains the velocity of 82 galaxies in the Corona Borealis region. To support a particular theory about the formation of galaxies, the analysis aims to estimate the number of stellar populations. This is a benchmark data set, well investigated in the literature, for example in \cite{EscobWest:1995}, \cite{RichGreen:1997} and \cite{Grazianetal:2019}, among others. We consider the galaxies velocities as random variables distributed according to a mixture of $k$ normal densities. The estimation of the number of components has proved to be delicate, with estimates ranging from 3 to 7, depending on factors, such as the priors for the parameters and the Bayesian method used. The histogram of the data set with a superimposed density is presented in Fig.~\ref{fig:galaxy}.

\begin{figure}[htb]
\centering
\includegraphics[width=14cm,height=6cm]{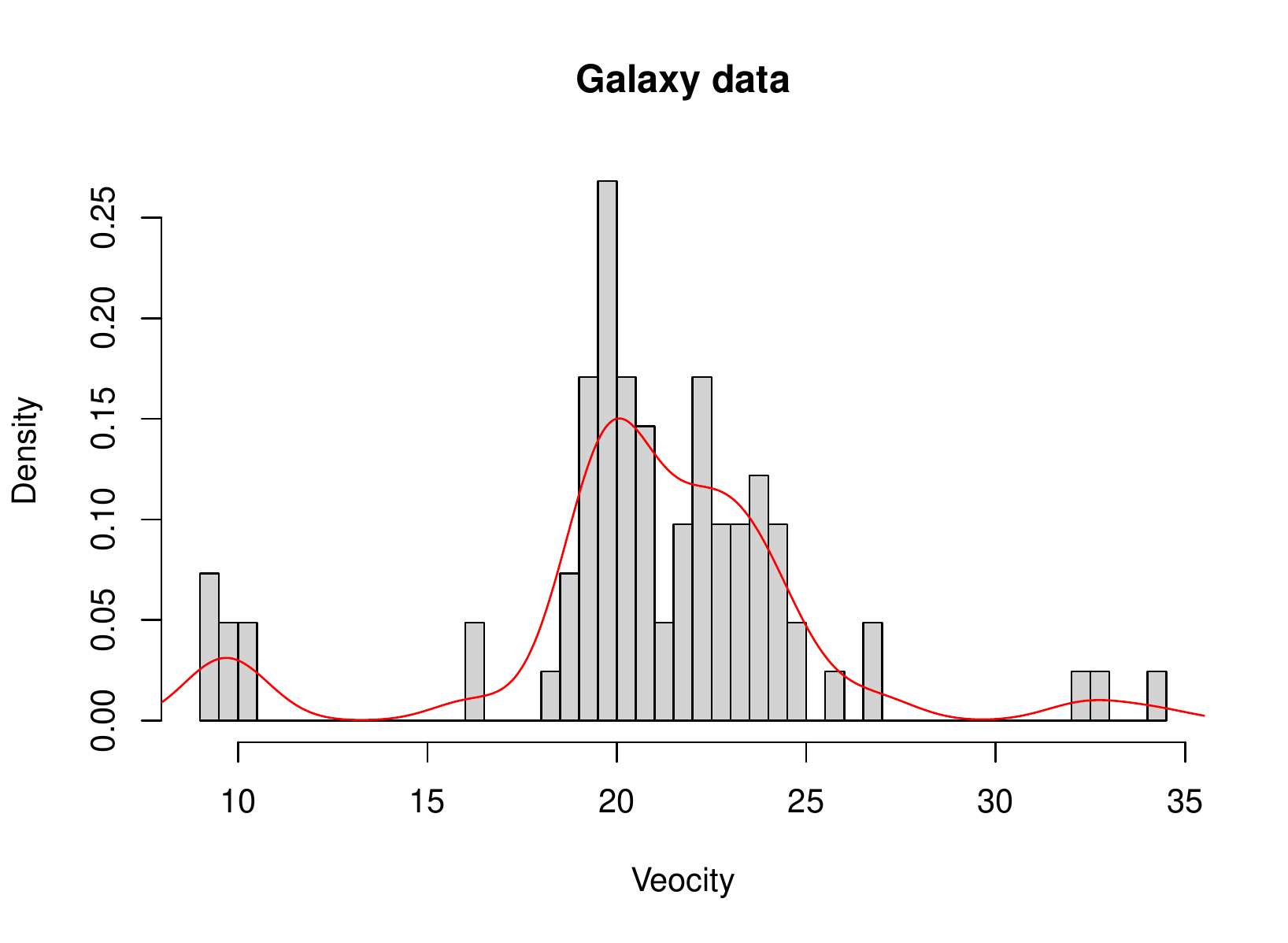}
\caption{Histogram of the Galaxy data set with a smoothed density (red curve) superimposed.}
\label{fig:galaxy}
\end{figure}

To select the number of components in the mixture of normal densities, we have fitted models with $k=(2,3,4,5,6,7,8)$ components, computing the Deviance Information Criterion (DIC) under each model. The results are reported in Table \ref{tab:galaxy}. We notice that, according to the computed index, we identify as best model the mixture with $k=4$ components, which is in line with  \cite{Grazianetal:2019}, and slightly more conservative than \cite{RichGreen:1997} and \cite{Grazianetal:2019}, where the number of components with non-zero weight is 5. Table \ref{tab:galaxyestimates} shows the posterior means for the parameters of the 4 components estimated.

\begin{table}[htb]
\centering
\begin{tabular}{c|ccccccc}
\hline 
$k$ & 2 & 3 & 4 & 5 & 6 & 7 & 8 \\ 
DIC & 476.72 & 425.20 & 371.48 & 413.68 & 446.45 & 446.54 & 458.81 \\ 
\hline 
\end{tabular} 
\caption{Deviance Information Criterion for mixture with number of components $k=(2,3,4,5,6,7,8)$.}
\label{tab:galaxy}
\end{table}

\begin{table}[htb]
\centering
\begin{tabular}{c|ccc}
\hline 
Component & $\omega$ & $\mu$ & $\sigma$ \\ 
\hline 
1 & 0.32 & 19.16 & 0.71 \\ 
2 & 0.31 & 23.59 & 3.46 \\ 
3 & 0.29 & 22.01 & 4.59 \\ 
4 & 0.10 & 9.31 & 0.58 \\ 
\hline 
\end{tabular} 
\caption{Posterior means of the parameters for the $k=5$ components. The results are reported by descending order of the weights.}
\label{tab:galaxyestimates}
\end{table}

\section{Variance parameters in hierarchical models}\label{sc_hierarchical}
In this section we discuss a well-known implementation of a hierarchical model that is proposed, for example,  in \cite{Gelman:2006}. The basic two-level hierarchical model is as follows:
\begin{eqnarray*}
y_{ij}&\sim& N(\mu+\alpha_j,\sigma_y^2),  \qquad i=1,\ldots,n,\,\,j=1,\ldots,J\\
\alpha_j&\sim& N(0,\sigma_\alpha^2), \qquad j=1,\ldots,J.
\end{eqnarray*}
This model has three parameters, namely $\mu$, $\sigma_y$ and $\sigma_\alpha$. However, out interest for this paper is in $\sigma_\alpha$ only, noting that ``regular'' objective priors can be used on the remaining parameters, like $\pi(\mu,\sigma_y)\propto1$, for example. Although being improper, this prior yields a proper posterior on the parameters.

The actual concern is on the variance (scale) parameter $\sigma_\alpha$, as if we were to put an improper prior on it, then the corresponding posterior, most likely, would be improper as well.  To compare the proposed prior, we assign an inverse-gamma prior on the variance with parameters set so to define a very sparse probability distribution. This is recommended, for example, in \cite{Spieg:1994}, where the prior is $\pi(\sigma_\alpha^2)\sim IG(\varepsilon,\varepsilon)$, with $\varepsilon>0$ sufficiently small. We do not discuss in detail the appropriateness of the above choice, or other alternatives; the reader can refer to \cite{Gelman:2006}, for example, for a through discussion.

The data consists of $J=8$ educational testing experiments, where the parameters $\alpha_1,\ldots,\alpha_8$ represent the relative effects of Scholastic Aptitude Test coaching programs in different schools. In this example, the parameter $\sigma_\alpha$ represents the between-schools variability (standard deviation) of the effects.  Table \ref{Tab:schools} shows the data.

\begin{table}[h]
\centering
\begin{tabular}{c|cc}
\hline 
School & $y_j$ & $\sigma_j$ \\ 
\hline 
A & 28 & 15 \\ 
B & 8 & 10 \\ 
C & -3 & 16 \\ 
D & 7 & 11 \\ 
E & -1 & 9 \\ 
F & 1 & 11 \\ 
G & 18 & 10 \\ 
H & 12 & 18 \\ 
\hline 
\end{tabular}
\caption{Observed effects $(y_j)$ of special preparation on SAT scores on eight randomised experiments, where $\sigma_j$ are the standard errors of effect estimate.}
\label{Tab:schools} 
\end{table}

We have compared the marginal posteriors $\pi(\sigma_\alpha^2|\mathbf{y})$ obtained by using the inverse-gamma prior with $\varepsilon=1$ and the proposed prior in \eqref{eq_prior2} with $a=1$. The histograms of the marginal posteriors are in Figure \ref{fig:schools}, where we note similar results. The statistics of the posteriors are reported in Table \ref{tab:schools_poststas}, where we note a less-informative distribution when the proposed prior is employed. This is expected, as the inverse-gamma distribution is considered a relatively informative one \citep{Gelman:2006}.

\begin{figure}[h]
\centering
\begin{subfigure}{0.45\textwidth}
  \includegraphics[width=\linewidth,height=6cm]{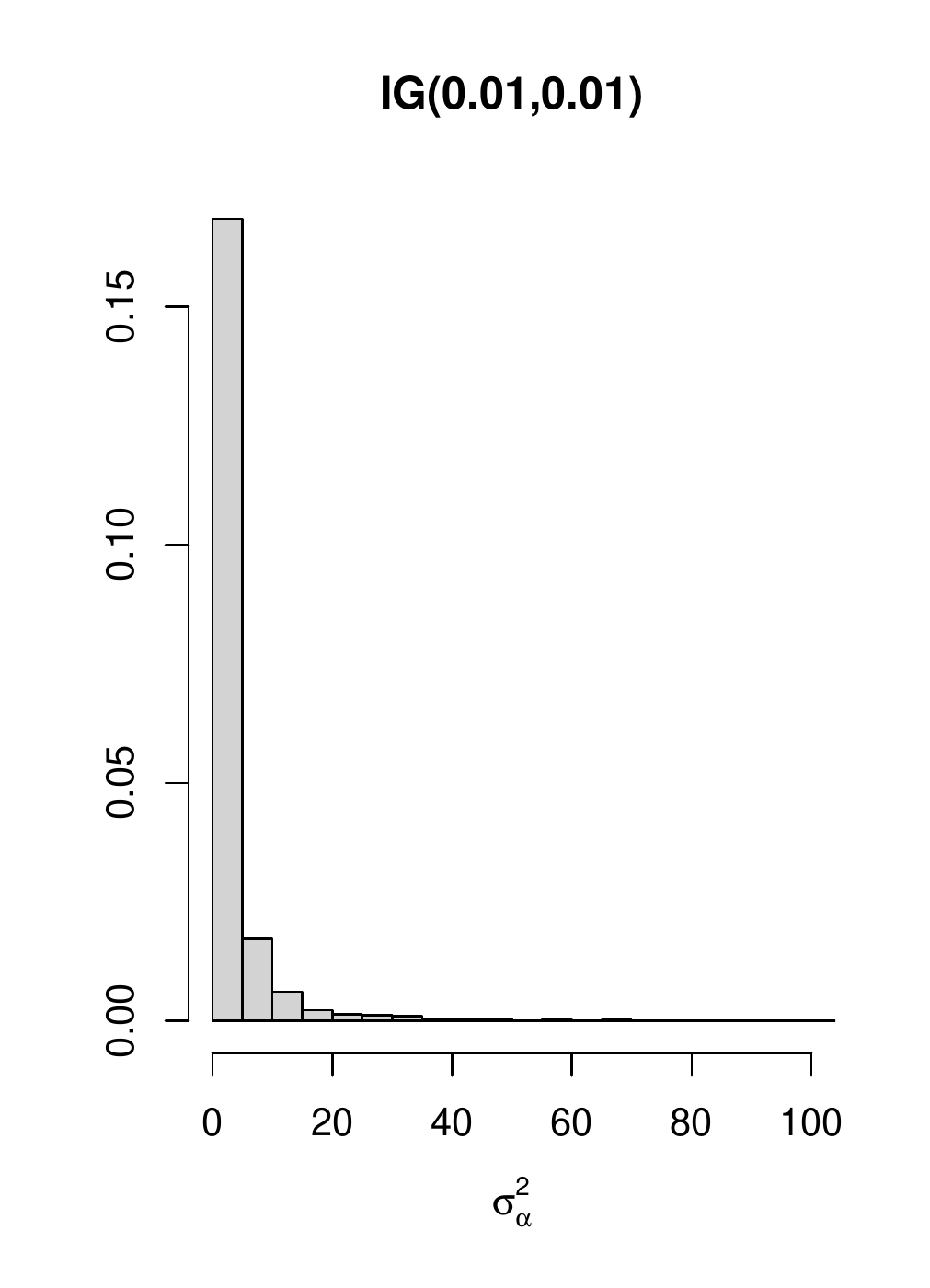}
  \label{fig:1a}
\end{subfigure}
\begin{subfigure}{0.45\textwidth}
  \includegraphics[width=\linewidth,height=6cm]{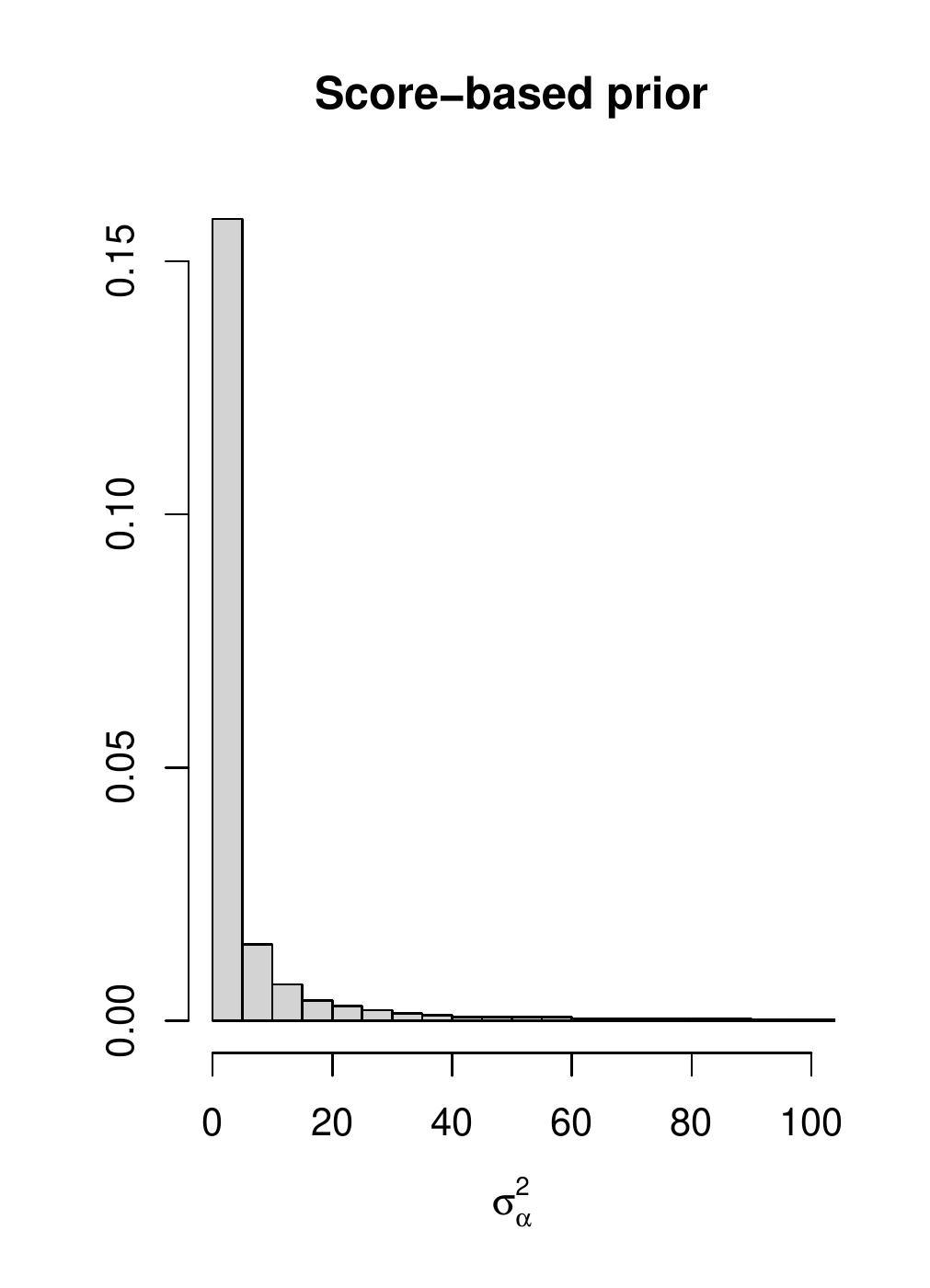}
  \label{fig:2a}
\end{subfigure}
\caption{Histogram of the marginal posterior for $\sigma_\alpha^2$ for the School problem when using an inverse-gamma prior (left) and the proposed prior based on scores (right).}
\label{fig:schools}
\end{figure}

\begin{table}[h]
\centering
\begin{tabular}{c|cc}
\hline 
Prior & Mean & 95\% C.I. \\ 
\hline 
Inverse-Gamma & 3.8 & (0.3, 23.9) \\ 
Score-based & 2.3 & (0.2, 57.8) \\ 
\hline 
\end{tabular}
\caption{Posterior statistics for the marginal distribution of $\sigma_\alpha^2$ for the schools problem.}
\label{tab:schools_poststas}
\end{table}

\section{Discussion}\label{sc_discussion}
In this paper we have derived a class of objective prior distributions that have the appealing properties of being proper and heavy-tailed. These have been obtained by exploiting a straightforward approach to the construction of score functions (here proposed). In detail, using convex function $\alpha(\cdot)$ we can find the score function with first two derivatives using \eqref{ascore}.  The Hyv\"{a}rinen score arises with $\alpha(u)=u^2$; whereas we have used $\alpha(u)=u^{-2}$ and used it to construct objective prior distributions using methodology  introduced in \cite{LVW:2020}.

The class of prior is heavy-tailed, behaving as $1/x^2$ for large $|x|$; this result is immeditely obvious as the prior on $(0,\infty)$ is a Lomax distribution with shape 1. In this respect, it behaves similar to standard objective priors but comes without the problems of being improper. The benefits of using a proper prior is that the posterior is automatically proper and so does not need to be checked.

We have showed that, when compared to Jeffreys prior on simulated data, the frequentist performances of the prior distribution derived from score functions are nearly equivalent. In addition, we have showed that, on both simulated and real data, the proposed prior is suitable to be used in a key scenario where improper priors (e.g. Jeffreys and reference) are not suitable (or are yet to be found). We have also illustrated the prior on a common problem for hierarchical models, that is assigning an objective prior for the variance parameter. 

As a final point, we briefly discuss the case where a prior is needed on a multidimensional parameter space. So, say we have a model with $k$ parameters, that is $\pmb\theta=(\theta_1,\ldots,\theta_k)$, where $\theta_j\in\Theta_j$, for $j=1,\ldots,k$. We also assume that the uni-dimensional space for each parameter is either $(0,\infty)$ or $(-\infty,\infty)$. Assuming $k$ relatively large, besides some specific statistical models such as regression models or graphical models, a common practice to assign objective priors on $\pmb\theta$ is as follows:
$$\pi(\pmb\theta) = \pi_1(\theta_1)\times\cdots\times\pi_k(\theta_k).$$
In other words, parameters are assumed to be independent a priori, so the join prior distribution is represented by the product of the marginal priors on each parameter. We can then set $\pi_j(\theta_j)$ to be either \eqref{eq_prior2} or \eqref{eq:cases}, for $j=1,\ldots,k$, depending on $\Theta_j$.


\appendix
\section{Appendix} 

\subsection{Information and Bregman divergence}\label{sc_informationandbregman}

Information for a density function is assumed to be convex, based on the obvious notion that averaging reduces information \citep{Topsoe:2001}. Hence, if we write $I(p)=\int \phi(p)$, then $\phi$ is assumed convex. 
For example, when $\phi(p)=p\log p$, $\phi$ is convex, as well as when $\phi(p,p')=(p')^2/p$, being convex in $(p,p')$.

With a convex function we can set up a Bregman divergence. First,  recall that a Bregman divergence (or Bregman distance) is a measure of the distance between two points, and it is defined in terms of a convex function. In particular, when the two points are probability distributions, then we have a statistical distance. Probably the most elementary Bregman distance is the (squared) Euclidean distance.

Let us start by considering the case $m=2$, so $\phi$ is a function of $(p,p')$.
Let $u$ and $v$ be vectors in $\mathbb{R}^d$. The Bregman divergence with convex function $\phi:\mathbb{R}^d\to\mathbb{R}$ is given by
$$B_\phi(u,v)=\phi(u)-\phi(v)-\langle \nabla \phi(v),u-v\rangle,$$
where the right most term is
$$\sum_{j=1}^d \frac{\partial \phi(v)}{\partial v_j}(u_j-v_j).$$
So $B_\phi(u,v)\geq 0$, and equality holds if and only if $u=v$.

In this work, we will focus on such divergences where the arguments are probability density functions \citep{StumVaj:2012}. Hence, in the one dimensional case, with $\phi:\mathbb{R}_+\to\mathbb{R}$ a convex function and $p(x)$ a density function, define
\begin{equation}\label{onebreg}
B_\phi(p(x),q(x))=\phi(p(x))-\phi(q(x))-\frac{\partial \phi(q)}{\partial q}(x)\,\,(p(x)-q(x)),
\end{equation}
and the divergence between $p$ and $q$ defined as
\begin{equation}\label{bdiv}
D(p,q)=\int B_\phi(p(x),q(x))\,\d x.
\end{equation}
For example, if $\phi(p)=p\,\log p$, which is convex in $p$, then
\begin{eqnarray*}
D(p,q) &=& \int p\log p-\int q\log q-\int (1+\log q)\,(p-q) \\
&=& \int p\,\log(p/q),
\end{eqnarray*}
which is the well known Kullback--Leibler divergence. The $L_2$ divergence arises when $\phi(p)=p^2$, with $B_\phi(p,q)=(p-q)^2$.

Note that we can write
$$D(p,q)=\int \phi(p)-\int p\left[\frac{\partial \phi}{\partial q}-\int \left\{q\frac{\partial \phi}{\partial q}-\phi(q) \right\}\right],$$
and so we see that
$$S(q)=\frac{\partial \phi}{\partial q}-\int \left\{q\frac{\partial \phi}{\partial q}-\phi(q) \right\}$$
is a score function, the Bregman score, and is local if $\phi$ satisfies $q\,\partial \phi/\partial q=\phi(q)$.
In this case the only solution is the log-score. As mentioned above, the log-score is the sole local and proper score function, that is depending on $q$ only. So, the Bregman divergence and the Bregman score confirm this, together with the results in \cite{Bern:1979}. We will see this is also the case with order $m=2$.  

The extension to the Bregman divergence and the Bregman score that we discuss in this paper, follows from a two-dimensional Bregman divergence with arguments $(p,p')$, where $p'(x)=dp/dx$, and for some two-dimensional convex function $\phi(u,v)$. It is defined by
\begin{equation}\label{breg2dim}
B_\phi(p(x),q(x))=\phi(p(x),p'(x))-\phi(q(x),q'(x))-\frac{\partial \phi}{\partial q}(x)\,(p(x)-q(x))-
\frac{\partial \phi}{\partial q'}(x)\,(p'(x)-q'(x)).
\end{equation}
For example, if $\phi(u,v)=v^2/u$, which is easily shown to be convex, we get
$$D(p,q)=\int p\left(\frac{p'}{p}-\frac{q'}{q}\right)^2,$$
known as the Fisher information divergence; see, for example, \cite{Villani:2008}.\\

\subsection{Divergences and scores}\label{sc_divergences}

The two most well known measures of information are differential Shannon and Fisher. See, for example, \cite{MacKay:2003}. Associated divergences are the Kullback--Leibler and Fisher, respectively, with corresponding score functions the logarithmic and Hyv\"{a}rinnen. The connection between divergence, information and score functions can be understood from the following;
\begin{equation}\label{divinfsc}
D(p,q)=I(p)+\int p \,S(q).
\end{equation}
Here $D$ denotes divergence, $I$ information and $S$ score. From this it is clear that $\int pS(q)$ is minimized at $q=p$. For \eqref{divinfsc} based on the  Kullback--Leibler divergence, we have
$$D(p,q)=\int p\,\log(p/q),\quad I(p)=\int p\log p\quad\mbox{and}\quad S(q)=-\log q.$$
For \eqref{divinfsc} based on Fisher information, we have
\begin{equation}\label{gendiv}
D(p,q)=\int p\left(p'/p-q'/q\right)^2,\quad I(p)=\int (p')^2/p\quad\mbox{and}\quad S(q)=2q''/q-(q'/q)^2.
\end{equation}
From \eqref{onebreg} we get
$$D(p,q)=\int \phi(p)+\int p\left\{-\frac{\partial \phi}{\partial q}+\int \left[q\,\frac{\partial \phi}{\partial q}-\phi(q)\right]\right\},$$
so 
$$S(q)=-\frac{\partial \phi}{\partial q}+\int \left[q\,\frac{\partial \phi}{\partial q}-\phi(q)\right].$$
More generally, using the two dimensional Bregman divergence in \eqref{breg2dim}, and put it in \eqref{bdiv}, we obtain
$$D(p,q)=\int \phi(p,p')+\int p\,S(q,q',q''),
$$
where the score is
\begin{equation}\label{scores}
S(q,q',q'')=-\frac{\partial\phi}{\partial q}+\frac{d}{dx}\frac{\partial \phi}{\partial q'}-\left\{\int \phi(q,q')-\int q\,\frac{\partial \phi}{\partial q}-\int q'\frac{\partial\phi}{\partial q'}\right\}.
\end{equation}
The score $S(q)$ has been obtained by implementing integration by parts, and assuming $[p\cdot \partial \phi/\partial q']$ vanishes at the boundary points. To ensure the score is local, we use the following condition on $\phi$,
\begin{equation}\label{condphi}
\phi(u,v)=u\,\frac{\partial \phi}{\partial u}+v\,\frac{\partial\phi}{\partial v}.
\end{equation}
Hence, the proposed class of score functions, which effectively is the class introduced in \cite{Parry:2012}, is given by 
\begin{equation}\label{score}
S(q,q',q'')=-\frac{\partial\phi}{\partial q}+\frac{d}{dx}\frac{\partial \phi}{\partial q'}.
\end{equation}

The derivation just presented is arguably much simpler to the one used in \cite{Parry:2012}, as we have avoided any variational analysis and the use of differential operators.

For a specific example, consider the class of convex functions, satisfying \eqref{condphi}, given by
\begin{equation}\label{diva}
\phi(u,v)=u\,\alpha(v/u),
\end{equation}
for some convex function $\alpha:\mathbb{R}\to\mathbb{R}$.
The convexity of $\alpha$ implies the convexity of function $\phi$ in \eqref{diva}, as the following Lemma \ref{lem_convex} shows.

\begin{lemma}\label{lem_convex}
The function $\phi(u,v)=u\alpha(v/u)$ is convex when $\alpha$ is convex.
\end{lemma}

\begin{proof}
The Hessian matrix corresponding to $\phi$ is seen to be
$$
\frac{\alpha''(v/u)}{u}\left(
\begin{matrix}
(v/u)^2  & v/u \\ \\
-v/u & 1
\end{matrix}
\right),
$$
which can be shown to be positive definite when $\alpha''>0$. Given that $\alpha$ is assumed to be convex, then the condition $\alpha''>0$ is true.
\end{proof}

\vspace{0.2in}
\noindent
Given the above form for $\alpha$, we are now able to find the divergence, the information and the score. However, before proceeding, we first note that
$$-\phi(q,q')+q\,\frac{\partial \phi(q,q')}{\partial q}+q'\,\frac{\partial \phi(q,q')}{\partial q'}=0.$$
That is, $\phi$ satisfies the condition in \eqref{condphi}.
In fact, we have
\begin{equation}\label{cond}
-u\,\alpha(v/u)+u\,\left[\alpha(v/u)+\alpha'(v/u)\,(-v/u^2)\right]+v\,\alpha'(v/u)\,(1/u)=0
\end{equation}
for all $(u,v)$.

\subsection{Higher order score functions}\label{sc_highorder}

In this section we look at score functions using an arbitrary number of derivatives; that is, we allow $m>2$.
So let $\phi$ now be a convex function on $(m+1)$ dimensions:
$$\phi(u_0,\ldots,u_{m})\geq \phi(v_0,\ldots,v_{m})+\sum_{j=0}^{m}\frac{\partial \phi}{\partial v_j}(u_j-v_j).$$
The Bregman divergence can be written as
$$B_\phi(q,p)=\phi(q_0,q_1,\ldots,q_m)-\phi(p_0,p_1,\ldots,p_m)-\sum_{j=0}^m \frac{\partial \phi}{\partial p_j}(q_j-p_j),$$
where the subscript $j$ indicates the order of differentiation, with, for example, $p_0=p$ and $p_j=d^j p/d x^j$. Also, we have 
$D(p,q)=\int B_\phi(p,q).$
In this Section, to keep a readable notation, we set $p=(p_0,p_1,\ldots,p_m)$ and $q=(q_0,q_1,\ldots,q_m)$.
If we have
\begin{equation}\label{const}
\phi(p)=\sum_{j=0}^m  \frac{\partial \phi}{\partial p_j}p_j,
\end{equation}
and the derivatives disappear at boundary values, using multiple integration by parts we get
$$
\begin{array}{ll}
D(q,p)  &  =\int \phi(p)-\sum_{j=0}^m \int p_j\,\frac{\partial\phi}{\partial q_j} \\ 
& =\int \phi(p)-\int p\,\sum_{j=0}^m (-1)^j \frac{d^j}{d x^j} \frac{\partial\phi}{\partial q_j}.
\end{array} 
$$
Hence, the score function is given by
\begin{equation}\label{mscore}
S(q)=\sum_{j=0}^m (-1)^{j+1}\,\frac{d^j}{d x^j} \frac{\partial\phi}{\partial q_j}.
\end{equation}
Note that, if we have $m=0$ and $\phi(u)=u\log u-u$, we recover the log-score function, that is $S(q)=-\log q$. If we have $m=1$ and $\phi(u,v)=v^2/u$, we recover the Hyv\"{a}rinen score function.

A general form of convex function satisfying the requirement \eqref{const} is
$$\phi(u_0,\ldots,u_{m})=\sum_{j=0}^{m-1} u_j\,\alpha_j(u_{j+1}/u_j),$$
where $(\alpha_0,\alpha_1,\ldots,\alpha_{m-1})$ is a set of convex functions.


\subsection{Connection with \cite{Parry:2012}}\label{ParryApp}

In \ref{sc_divergences}, we have shown that it is possible to derive the class of score functions using convex functions and the Bregman divergence only. Here we show that the properties we have used, imply those used by \cite{Parry:2012}.

First we show that a function $\phi$ satisfying \eqref{condphi} implies that $\phi$ is $1$--homogeneous.

\begin{lemma}\label{lemonehom}
Let $\phi$ be such that $\phi(u,v)=u\partial\phi/\partial u+v\partial \phi/\partial v$. Then, for any $\lambda>0$, it is that
$\phi(\lambda u,\lambda v)=\lambda\,\phi(u,v)$.
\end{lemma} 

\begin{proof}
Let
$T(\lambda)=\phi(\lambda u,\lambda v)$. Thus, $\partial T/\partial\lambda=T(\lambda)/\lambda$, which implies $T(\lambda)\propto \lambda$.
\end{proof}

\noindent
Next we show that the class of score functions \eqref{score} and the condition \eqref{condphi}, imply \eqref{euler}.

\begin{lemma}\label{lemeuler}
If a function $\phi$ satisfies the condition in \eqref{condphi}, and a score $S$ is given by \eqref{score}, then the score $S$ satisfies equation \eqref{euler}. 
\end{lemma}

\begin{proof}
The proof is a simple matter of algebra and calculus. It requires the key observation that
$$\frac{\partial }{\partial q}\frac{d}{dx}=\frac{d}{dx}\frac{\partial}{\partial q}
\quad\quad\mbox{and}\quad\quad
\frac{\partial }{\partial q'}\frac{d}{dx}=\frac{d}{dx}\frac{\partial}{\partial q'}+\frac{\partial}{\partial q}.
$$
Given that $\phi=q\partial \phi/\partial q+q'\partial \phi/\partial q'$ we get
$$q\frac{\partial^2\phi}{\partial q^2}+q'\frac{\partial^2\phi}{\partial q\partial q'}=0\quad\mbox{and}\quad
q\frac{\partial^2\phi}{\partial q\partial q'}+q'\frac{\partial^2\phi}{\partial q'^2}=0.
$$
Since $S=-\partial\phi/\partial q+q'\partial^2\phi/\partial q\partial q'+q''\partial^2\phi/\partial q'^2$, we have
$\partial S/\partial q'' =\partial^2\phi/\partial q'^2$. Further,
$$\frac{\partial S}{\partial q'}=-\frac{\partial^2\phi}{\partial q\partial q'}+\frac{d}{dx}\frac{\partial^2\phi}{\partial q'^2}+\frac{\partial^2\phi}{\partial q\partial q'}\quad
\mbox{and}\quad\frac{\partial S}{\partial q}=-\frac{\partial^2\phi }{\partial q^2}+\frac{d}{dx}\frac{\partial^2\phi}{\partial q\partial q'}.$$
Combining all the above expressions, yields \eqref{euler}.
\end{proof}

\noindent
Finally, in this Appendix, we show that if $\phi$ is $1$--homogeneous, then $\phi$ is convex. The result is quite straightforward, and is achieved by showing that $\partial^2\phi/\partial u^2\,\cdot \partial^2\phi/\partial v^2\geq (\partial^2\phi/\partial u\partial v)^2$.

Here we make the connection with equation (39) in \cite{Parry:2012} and the score function in \eqref{mscore}. The main result is to show that if 
$$\phi(u)=\sum_{j=0}^m u_j\,\frac{\partial \phi}{\partial u_j}\quad \mbox{and}\quad S(u)=\sum_{j=0}^m (-1)^{j+1}\frac{d^j}{dx^j}\frac{\partial \phi}{\partial u_j},$$
then the relevant Euler--Lagrange equation is
$$\sum_{j=0}^{2m}(-1)^j\frac{d^j}{dx^j}\left(u_0\frac{\partial S}{\partial u_j}\right)=0.$$
To this end, define the operators, as in \cite{Parry:2012}
$$E=\sum_{j} u_j\,\frac{\partial}{\partial u_j},\quad \Lambda=\sum_{j} (-1)^{j+1}\frac{d^j}{dx^j}\frac{\partial}{\partial u_j}\quad\mbox{and}\quad L=\sum_{j}(-1)^j\frac{d^j}{dx^j}\left(u_0\frac{\partial}{\partial u_j}\right).$$

\begin{theorem}
If $S=\Lambda \phi$ and $E\phi=\phi$, then $LS=0$.
\end{theorem}

\begin{proof}
To proof the theorem, we consider the properties of differential operations, as discussedi in Section 5 of \cite{Parry:2012}. In particular, the properties $\Lambda \,E=E\Lambda+\Lambda$ and $E\,L=L\,E$.

\noindent
Hence, $L\Lambda \phi=L\,\Lambda\,E\phi=L(E\Lambda+\Lambda)\phi$, which implies that $L\,E\,\Lambda\phi=0$.
This in turn implies $E\,L\Lambda\phi=0$; i.e. $E\,LS=0$. Finally, it is easy to see that if $E\psi=0$ then $\psi=0$.
\end{proof}


\end{document}